\RequirePackage[l2tabu, orthodox]{nag}
\documentclass{article}

\usepackage[T1]{fontenc}

\usepackage[bitstream-charter]{mathdesign}
\usepackage{bbm}
\usepackage{amsmath}
\usepackage[scaled=0.92]{PTSans}
\usepackage{csquotes}

\usepackage[
  paper  = letterpaper,
  left   = 1.2in,
  right  = 1.2in,
  top    = 1.0in,
  bottom = 1.0in,
  ]{geometry}

\usepackage[dvipsnames]{xcolor}
\definecolor{shadecolor}{gray}{0.9}

\usepackage[final,expansion=alltext]{microtype}
\usepackage[parfill]{parskip}
\usepackage[english]{babel}
\usepackage{setspace}
\usepackage{graphicx}
\usepackage[labelfont=bf]{caption}
\usepackage[format=hang]{subcaption}

\usepackage{booktabs}
\usepackage{tabularx}
\usepackage{multirow}

\usepackage{authblk}

\usepackage[round, sort]{natbib}

\usepackage[colorlinks,linktoc=all]{hyperref}
\usepackage[all]{hypcap}
\hypersetup{citecolor=BrickRed}
\hypersetup{linkcolor=black}
\hypersetup{urlcolor=MidnightBlue}

\usepackage{tcolorbox}

\usepackage{listings}
\usepackage{algorithm}
\usepackage{algpseudocode}

\usepackage{amsthm}
\newtheorem{theorem}{Theorem}[section]
\newtheorem{lemma}[theorem]{Lemma}
\newtheorem{proposition}[theorem]{Proposition}

\theoremstyle{definition}

\newtheorem{remark}[theorem]{Remark}
\newtheorem{assumption}[theorem]{Assumption}

\usepackage{enumitem}

\usepackage{tikz} 

\usepackage{pifont}

\newcommand{\RR}{\mathbbmss{R}} 
\newcommand{\EE}[1]{\mathbbmss{E}\left[#1\right]} 
\newcommand{\PP}[1]{\mathbbmss{P}\left[#1\right]} 
\newcommand{\Normal}[2]{\mathcal{N}\left(#1,#2\right)} 
\newcommand{\pto}{\overset{p}{\rightarrow}} 
\newcommand{\dto}{\leadsto} 
\newcommand{\Lto}[1]{\overset{\mathbbmss{L}^{#1}}{\rightarrow}} 
\newcommand{\influence}[1]{\varphi\left(#1\right)}

\newcommand{\data}{\mathcal{D}}

\newcommand{\indep}{\perp \!\!\!\! \perp}
\newcommand{\onea}[1]{\mathbbmss1\left\{#1\right\}} 
\newcommand{\Ltwo}{\mathbbmss{L}_2} 
\newcommand{\norm}[1]{\left\|#1\right\|_{2}} 

\newcommand{\thetatarget}{\theta^\star}

\newcommand{\pitarget}{\pi^\star}
\newcommand{\mutarget}{\mu^\star}

\newcommand{\clip}[1]{\textnormal{clip}\left( #1 \right)}

\newcommand{\jPtarget}{\mathbbmss{P}^\star}
\newcommand{\finfluence}[1]{\varphi^{F}\left(#1\right)}

\newcommand{\ACC}{\texttt{DR+ACC}}
\newcommand{\oreg}{\hat\theta_{\texttt{OR}}}
\newcommand{\ipw}{\hat\theta_{\texttt{IPW}}}
\newcommand{\dr}{\hat\theta_{\texttt{DR}}}
\newcommand{\acc}{\hat\theta_{\ACC}}

\newcommand{\risk}[1]{\mathcal{R}\left[#1\right]}

\newcommand{\shortfif}{\varphi^F}

\newcommand{\zoreg}{Z_\texttt{OR}}
\newcommand{\zipw}{Z_\texttt{IPW}}
\newcommand{\zcorr}{Z_\textnormal{correction}}

\title{\textbf{Rescuing double robustness: safe estimation under complete misspecification}}

\author[1,2]{Lorenzo Testa}
\author[2,3]{Francesca Chiaromonte}
\author[1,4]{Kathryn Roeder}

\affil[1]{Department of Statistics \& Data Science, Carnegie Mellon University, Pittsburgh PA, US}
\affil[2]{L'EMbeDS, Sant'Anna School of Advanced Studies, Pisa, Italy}
\affil[3]{Department of Statistics, Penn State University, University Park PA, US}
\affil[4]{Department of Computational Biology, Carnegie Mellon University, Pittsburgh PA, US}

\date{\vspace{-0.1em}\texttt{lorenzo@stat.cmu.edu} \quad \texttt{fxc11@psu.edu} \quad \texttt{roeder@andrew.cmu.edu}\\ \vspace{1em}\today}

\begin{document}

\maketitle

\begin{abstract}
    Double robustness is a major selling point of semiparametric and missing data methodology. Its virtues lie in protection against partial nuisance misspecification and asymptotic semiparametric efficiency under correct nuisance specification. However, in many applications, complete nuisance misspecification should be regarded as the norm (or at the very least the expected default), and thus doubly robust estimators may behave fragilely. In fact, it has been amply verified empirically that these estimators can perform poorly when all nuisance functions are misspecified. Here, we first characterize this phenomenon of \textit{double fragility}, and then propose a solution based on \textit{adaptive correction clipping} (\ACC). We argue that our \ACC{} proposal is \textit{safe}, in that it inherits the favorable properties of doubly robust estimators under correct nuisance specification, but its error is guaranteed to be bounded by a convex combination of the individual nuisance model errors, which prevents the instability caused by the compounding product of errors of doubly robust estimators. We also show that our proposal comes with no reduction in semiparametric efficiency compared to doubly robust estimators, and thus \textit{valid} inference based on asymptotic normality can be conducted when nuisances are well-specified. We showcase the efficacy of our \ACC{} estimator both through extensive simulations and by applying it to the analysis of Alzheimer's disease proteomics data.
\end{abstract}

\section{Introduction}
Scientific progress relies on stable and reproducible evidence, which in turn demands solid statistical methodology \citep{yu2020veridical, yu2024veridical}. In recent decades, a paradigm shift has occurred in various scientific fields, spurred by ideas and concepts from causal inference and missing data analysis, moving from a \enquote{bottom-up} approach, where statistical models are posited first and their parameters interpreted post-hoc, to a \enquote{top-down} approach \citep{kennedy2024semiparametric}. This framework begins by defining a specific target of inquiry -- the estimand -- and then systematically lays out the assumptions and methods required to identify and estimate it from data.

Within this top-down paradigm, doubly robust (\texttt{DR}) estimators have emerged as a particularly powerful and celebrated tool \citep{robins1994estimation, bang2005doubly, scharfstein1999adjusting}. Their appeal lies in two key properties. First, they offer protection against partial model misspecification: they yield a consistent estimate of the target parameter if at least one of the two required nuisance models -- typically an outcome regression model and a propensity score model -- is correctly specified. Second, if both nuisance models are correct, \texttt{DR} estimators achieve semiparametric efficiency, meaning they have the smallest possible asymptotic variance. These virtues have made them a cornerstone of semiparametric and missing data methodology.

However, the theoretical protection of double robustness can be misleading in practice. In applied research, it is often more realistic to assume that \textit{all} models are misspecified to some degree, rather than expecting one to be perfectly correct (even asymptotically). In this scenario of complete nuisance misspecification, the guarantees of \texttt{DR} estimators vanish. Worse, they can become fragile. It has been empirically observed that \texttt{DR} estimators can perform poorly in this setting, sometimes exhibiting more bias than simpler estimators that rely on only one of the misspecified nuisance functions \citep{kang2007demystifying, zhang2025doubly}. We term this phenomenon \textit{double fragility}: the very mechanism designed to provide robustness can, under complete misspecification, amplify estimation errors and lead to unstable results.

This paper provides a formal characterization of double fragility, focusing on the dual nature of the correction term in doubly robust estimators. When at least one of the nuisance models is correct, this term is beneficial, guiding the estimator toward the true parameter. Conversely, when both nuisance models are misspecified, we show that their errors can compound within the correction term, actively degrading the estimator performance rather than improving it. This analysis reveals the source of fragility and directly motivates our solution: \textit{adaptive correction clipping} (\ACC). Our proposal is built on the principle of \textit{safety} -- the property that an estimator performs no worse than the simpler estimators from which it is constructed \citep{xu2025unified}. By adaptively clipping the correction term, our \ACC{} method preserves the consistency and efficiency of standard \texttt{DR} estimators in ideal settings, while crucially preventing the correction from amplifying bias when all models are misspecified. We further show that adaptive correction clipping becomes inactive when nuisance models are well-specified, ensuring that our \ACC{} estimator converges to a normal distribution with optimal variance dictated by the semiparametric efficiency bound. Consequently, we demonstrate that \textit{valid} confidence intervals can be constructed using standard asymptotic normality.

To complement these theoretical and methodological results, we provide extensive empirical evidence of our \ACC{} estimator effectiveness. First, we conduct a simulation study that fully replicates the design of \citet{kang2007demystifying}, a well-known benchmark in the missing data literature. This setup is specifically designed to assess estimator performance across several scenarios of nuisance model specification, allowing for a rigorous evaluation of both fragility and safety. Second, we showcase the practical utility of our method by applying it to a substantive scientific problem: estimating the average treatment effect (ATE) of Alzheimer's disease on hundreds of peptide abundances, using data from \citet{merrihew2023peptide}. This application demonstrates the tangible benefits of our safe estimator in a real-world setting where the risk of model misspecification is high.

\subsection{A preview of our results}
Before diving into the details of our approach, we provide a brief sketch of what can go wrong with \texttt{DR}, and some intuition on our proposal. Assume that we observe a sample of $n$ independent and identically distributed random variables $\left\{\data_i =(X_i,R_i,R_iY_i) \right\}_{i=1}^{n}$. Here, $X_i$ is a vector of covariates, $R_i$ is a binary indicator equal to 1 if $Y_i$ is observed and 0 otherwise, and the outcome $Y_i$ is only observed when $R_i=1$. We let $\data = (X,R,RY)$ denote an independent copy of $\data_i = (X_i,R_i,R_iY_i)$. We operate under the standard assumptions of missing at random ($Y\indep R\mid X$) and positivity. Our goal is to estimate the target parameter 
\begin{equation}
    \thetatarget = \EE{Y} = \EE{\mutarget(X)} = \EE{\frac{RY}{\pitarget(X)}} = \EE{\mutarget(X) + \frac{R}{\pitarget(X)}(Y-\mutarget(X))}\,,
\end{equation}
where $\mutarget(x) = \EE{Y\mid X=x, R=1}$ denotes the nuisance regression function and $\pitarget(x) = \PP{R=1\mid X=x}$ denotes the nuisance propensity score. Assuming we have access to pre-trained or externally fitted models for these functions, denoted $\hat\mu$ and $\hat\pi$ respectively, the well-known \texttt{DR} estimator can be defined as
\begin{equation}
    \dr =  \frac{1}{n} \sum_{i=1}^n \left[\hat\mu(X_i) + \frac{R_i}{\hat\pi(X_i)}(Y_i - \hat\mu(X_i)) \right]\,
\end{equation}
The \texttt{DR} estimator is particularly compelling because it can be algebraically decomposed into:
\begin{equation}
    \dr = \underbrace{\frac{1}{n} \sum_{i=1}^n \hat\mu(X_i)}_{\oreg} + \underbrace{\frac{1}{n} \sum_{i=1}^n \frac{R_iY_i}{\hat\pi(X_i)}}_{\ipw} - \underbrace{\frac{1}{n} \sum_{i=1}^n  \frac{R_i\hat\mu(X_i)}{\hat\pi(X_i)}}_{\text{correction}}\,,
\end{equation}
which sheds light on the fact that the \texttt{DR} estimator is a combination of the simpler \textit{outcome regression} (\texttt{OR}) and \textit{inverse probability weighting} (\texttt{IPW}) estimators, respectively defined as
\begin{equation}
    \oreg =  \frac{1}{n} \sum_{i=1}^n \hat\mu(X_i)\,,\quad \ipw =  \frac{1}{n} \sum_{i=1}^n \frac{R_i Y_i}{\hat\pi(X_i)}\,.
\end{equation}
This simple algebraic structure has profound consequences on the property of the doubly robust estimator. Its correction term is engineered to ensure consistency and, ultimately, efficiency. This behavior can be seen in two key scenarios:
\begin{itemize}
    \item Under partial misspecification, the term acts as a safeguard. For instance, if the outcome model $\hat\mu$ is correct but the propensity model $\hat\pi$ is not, the correction term is designed to asymptotically cancel the biased \texttt{IPW} component. As a result, the \texttt{DR} estimator converges to the same correct limit as the \texttt{OR} estimator. A symmetric cancellation occurs if $\hat\pi$ is correct instead.
    \item Under full correctness, when both nuisance models are well-specified, all three components of the estimator -- \texttt{OR}, \texttt{IPW}, and the correction term -- converge to the true parameter. In this ideal case, the correction term is effectively free to cancel either the \texttt{OR} or the \texttt{IPW} component, a flexibility that leads not only to consistency but also to semiparametric efficiency.
\end{itemize}
While this property is widely known as double robustness, we argue it can be better understood as a form of \textit{asymptotic hard thresholding}, where the estimator effectively selects a valid component through the correction term. The first three panels of Figure \ref{fig:intro_plot_n1000} provide an empirical assessment of this principle. There, we show the distribution of the \texttt{OR}, \texttt{IPW} and \texttt{DR} estimators across 1000 replications (with $n=1000$) under scenarios where either both or at least one of the nuisance models are correct. As the panels illustrate, the \texttt{DR} estimator distribution aligns with that of the correctly specified component, thus ignoring the misspecified one. Although the full details of the simulation design are deferred to Section \ref{sec:sim}, these results offer a practical validation of the hard thresholding mechanism in action.

\begin{figure}[t!]
    \centering
    \includegraphics[width=\linewidth]{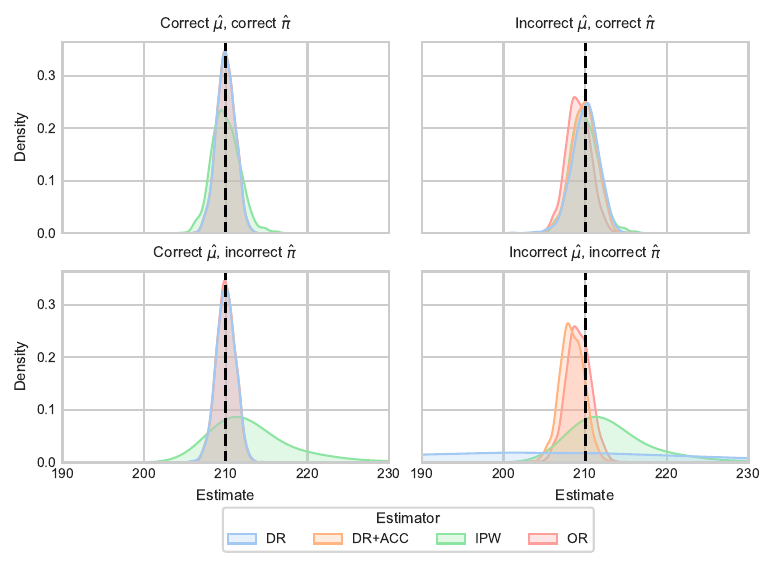}
    \caption{Sampling distributions for the Outcome Regression (\texttt{OR}), Inverse Probability Weighting (\texttt{IPW}), Doubly Robust (\texttt{DR}), and our proposed \ACC{} estimators from 1000 simulations with a sample size of $n=1000$. The true parameter value is 210. The four panels show the estimators' performance under all combinations of correct and incorrect nuisance model specifications. The top row and bottom-left panel demonstrate the \textit{asymptotic hard thresholding} property. When at least one nuisance model is correct, both \texttt{DR} and \ACC{} align with the correct simpler estimator. The bottom-right panel illustrates \textit{double fragility}. When both nuisance models are wrong, the bias of the standard \texttt{DR} estimator is substantially worse than that of either the \texttt{OR} or \texttt{IPW} estimators. In this challenging scenario, our proposed \ACC{} estimator is shown to be \textit{safe}, providing a much more stable and accurate estimate than the standard \texttt{DR} estimator. Full details of the simulation scenario are provided in Section \ref{sec:sim}.}
    \label{fig:intro_plot_n1000}
\end{figure}

The problem arises when both nuisance models are misspecified -- a common, if not ubiquitous, scenario in applied research. In this case, the asymptotic cancellation fails. The correction term, now a function of two incorrect models, can compound their errors in unpredictable ways, potentially introducing substantial bias that pushes the final estimate far from the true value. We call this failure mode \textit{double fragility}. Instead of reducing bias, the correction term can exacerbate it. This phenomenon is illustrated in the final panel of Figure \ref{fig:intro_plot_n1000}. Under complete model misspecification, the standard \texttt{DR} estimator is not only biased but also performs considerably worse than either of its simpler \texttt{OR} and \texttt{IPW} components.

To address this, we propose a simple yet effective solution: \textit{adaptive correction clipping} (\ACC). Our method works by constraining, or \enquote{clipping}, the correction term to ensure the final estimate is always anchored within the range defined by the simpler \texttt{OR} and \texttt{IPW} estimates. This modification enforces a safety property: it inherits the favorable properties of the standard \texttt{DR} estimator when its assumptions hold, but it is guaranteed to perform no worse than its constituent estimators when they fail. As Figure \ref{fig:intro_plot_n1000} demonstrates, the performance of our \ACC{} estimator is nearly identical to that of the standard \texttt{DR} when at least one nuisance model is correct. However, in the double fragility scenario, it substantially outperforms the standard \texttt{DR} estimator, validating its role as a safe and robust alternative.

\subsection{Related work}
Here, we discuss the relationship between our work and closely related scholarship.

\textbf{Semiparametric statistics and missing data.} From a technical standpoint, our work is grounded in the field of semiparametric statistics, which focuses on estimation in the presence of high-dimensional nuisance parameters. Missing data literature provides a rich toolkit for such problems, including methods such as Inverse Probability Weighting (\texttt{IPW}), augmented IPW, and targeted maximum likelihood estimation, many of which lead to doubly robust estimators \citep{horvitz1952generalization, hernan2006estimating, van2006targeted}. The foundational concepts of influence functions, tangent spaces, and semiparametric efficiency bounds -- developed in seminal works by \citet{bickel1993efficient, tsiatis2006semiparametric, van2000asymptotic} -- serve as the theoretical backbone for much of modern statistical learning. It is important to note, however, that most of these classical results are based on asymptotic arguments, which may not always hold in finite samples. Therefore, recent literature has increasingly focused on  finite-sample performance of semiparametric estimators. \citet{mou2022off} analyze two-stage procedures common in causal inference and prove non-asymptotic upper bounds on the mean-squared error, revealing that to achieve optimality in finite samples, the error in estimating the nuisance function should be minimized in a specific weighted $\Ltwo$-norm. Similarly, \citet{wang2024pac} provide a finite-sample analysis of doubly robust estimators using PAC-style guarantees, demonstrating that minimizing the estimation error of the treatment effect in terms of Chi-square distance is crucial for minimizing the final estimator variance. For the widely-used augmented inverse probability weighting (AIPW) estimator, \citet{wang2023non} explore its non-asymptotic properties, showing it can achieve near-oracle performance even when the nuisance models converge at a slow, non-parametric rate. This focus on non-asymptotic guarantees has also extended to specific, challenging scenarios. \citet{ghadiri2023finite} address the historical lack of non-asymptotic accuracy bounds for treatment effect estimation in the finite population setting, while \citet{celentano2023challenges} tackle the high-dimensional $n<p$ \enquote{inconsistency regime}, where they develop a novel procedure to achieve consistency when standard methods fail. Complementing these theoretical advances, empirical studies provide practical insights; for instance, \citet{witter2024benchmarking} used a novel benchmarking framework to show that simpler, doubly robust estimators often outperform more complicated methods in practice, a finding which in turn spurred new theory on the finite-sample variance of these estimators. More recently, \citet{colangelo2026double} develop doubly robust inference for continuous treatments under unconfoundedness, constructing estimators for the average dose-response function that are asymptotically normal at nonparametric rates.

\textbf{Stability, safety, and robust statistics.} The principles of stability and robustness are crucial for developing reliable and reproducible statistical methods. Classical robust statistics, with foundational work by \citet{hampel1974influence, huber1996robust, he1992reweighted, hampel2011robust}, introduced concepts like the influence function to create estimators that are insensitive to outliers or deviations from assumed data distributions. More recently, this idea has been broadened to the concept of stability, which encompasses the entire data science life cycle, emphasizing the importance of methods being resilient to perturbations in data, models, and analytical choices \citep{yu2020veridical, rewolinski2025pcs, agarwal2025pcs, yu2013stability, yu2024veridical}. The development of tools like the \textit{s-value} for evaluating stability against distributional shifts further highlights the field's focus on this property \citep{gupta2023s}. Our work connects directly to these developments by identifying a specific instability in doubly robust estimators, which we term double fragility. The safety property we propose is a specific form of stability tailored to this problem, guaranteeing that the estimator is robust against the complete misspecification of its underlying nuisance models \citep{deng2024optimal, xu2025unified}. We also want to note that the finite-sample instability of doubly robust estimators is a well-recognized challenge. A common \textit{ad hoc} adjustment is propensity score \textit{trimming}, where estimated probabilities of treatment are bounded away from 0 and 1 to prevent the inverse weights from becoming excessively large \citep{sturmer2021propensity, ma2023doubly}. Another approach is self-normalization, which is used in the Hájek estimator \citep{basu1971essay}. \citet{cai2024c} recently proposed the C-Learner, which reframes the construction of a debiased estimator as a constrained optimization problem. Instead of first training a nuisance model to optimize prediction accuracy and then applying a \textit{post hoc} correction, the C-Learner directly trains the outcome model to minimize prediction error subject to the constraint that the first-order error term is zero. 

\textbf{Semi-supervised learning and prediction-powered inference.} Our setup, which we will introduce shortly, mimics the one usually studied in semi-supervised learning. For instance, \citet{zhang2019semi, zhang2022high} have focused on semi-supervised mean estimation, including extensions to high-dimensional settings with bias-corrected inference. The problem of semi-supervised linear regression has also been explored in depth, leading to the development of asymptotically normal estimators with improved efficiency and minimax optimal estimators in high dimensions, with subsequent extensions to generalized linear models \citep{chakrabortty2018efficient, azriel2022semi}. The scope of semi-supervised has broadened to encompass more general inferential tasks like M-estimation and U-statistics \citep{chakrabortty2016robust, testa2025semiparametric, kim2024semi}. A closely related paradigm is prediction-powered inference (PPI), where an analyst leverages a pre-trained, black-box machine learning model in addition to labeled and unlabeled data \citep{angelopoulos2023prediction, angelopoulos2023ppi++, zrnic2024cross}.

\subsection{Roadmap}
The remainder of this paper is organized as follows. First, in Section \ref{sec:dr} we introduce our setup and notation, and then we analyze the behavior of standard doubly robust estimators, formally characterizing the \textit{double fragility} phenomenon that arises under complete nuisance model misspecification. Then, in Section \ref{sec:acc} we introduce our proposed solution, \ACC, an estimator that uses \textit{adaptive correction clipping} as a safeguard against this fragility. Here, we prove consistency, semiparametric efficiency, and safety of our estimator. In Section \ref{sec:sim} we provide empirical validation through an extensive simulation study, and in Section \ref{sec:app} we demonstrate the practical utility of our method with an application to the analysis of peptide abundance data in an Alzheimer's study. Section \ref{sec:end} contains some concluding remarks. Additional theoretical results and simulation details are provided in the Supplementary Material. All code for reproducing our analysis is available at \url{https://github.com/testalorenzo/DoubleFragility}.

\section{Problem setup and review of double robustness}
\label{sec:dr}

\subsection{Setup and notation}
Following traditional nomenclature from semiparametric statistics literature, we cast our problem in a missing data framework. We denote the \textit{observed data} as $\left\{ \data_i =(X_i,R_i,R_iY_i) \right\}_{i=1}^{n}$, where $R_i$ is a binary indicator that indicates whether observation $\data_i$ is labeled ($R_i=1$), or unlabeled ($R_i=0$). We let $\data = (X,R,RY)$ denote an independent copy of $\data_i = (X_i,R_i,R_iY_i)$. We denote the data-generating distribution as $\jPtarget\in\mathcal{P}$, where $\mathcal{P}$ is the set of distributions induced by a nonparametric model, and thus we write $\data\sim\jPtarget$. For theoretical convenience, we also define the \textit{full data} $\data^F = (X,Y)$, that is, the data that we would observe if there were no missingness mechanisms in place. 

The potential discrepancy between the labeled and unlabeled datasets due to \textit{distribution shift} naturally leads us to adopt a \textit{missing at random} (MAR) labeling mechanism. We formally state this assumption below, along with a weak overlap assumption -- 
as requirements for identifiability of the target parameter.

\begin{assumption}[Identifiability]
\label{ass:MAR}
    Let the following assumptions hold:
    \begin{enumerate}[label=\textbf{\alph*.}]
    \item \textbf{Missing at random.} $R \indep Y \mid X$. 
    \item \textbf{Weak overlap.} $\pi(x) = \PP{R=1\mid X=x} \in (0,1)$ for all $x\in\RR^p$ almost surely.
\end{enumerate}
    \end{assumption}

Throughout this paper, we focus on the estimation of a target quantity $\thetatarget\in\RR$ that is defined as the functional solving $\thetatarget = \theta(\jPtarget)$, with $\theta\,:\,\mathcal{P} \to \RR$.
We restrict our attention to target parameters that can be estimated using \textit{regular} and \textit{asymptotically linear observed-data} estimators, that is, targets that admit the expansion
\begin{equation}
\label{eq:ral}
    \sqrt{n} \left(\hat\theta - \thetatarget \right) = n^{-1/2} \sum_{i=1}^{n} \influence{\data_i;\thetatarget} + o_\mathbbmss{P}(1)\,,
\end{equation}
for some regular and asymptotically linear \textit{observed-data} estimator $\hat\theta$, and some function $\influence{\data_i;\thetatarget}$, referred to as \textit{observed-data influence function}, evaluating the contribution of the \textit{observed-data} sample $\data_i$ to the overall estimator $\hat\theta$ \citep{hampel1974influence}.
This framework is very general and encompasses M-estimation, Z-estimation, and many estimands in causal inference \citep{kennedy2024semiparametric, van2000asymptotic}.

In a nonparametric setting, semiparametric theory shows that for any given target parameter, there is a unique \textit{full-data influence function}. This theoretical building block can be projected onto the observed data to derive the \textit{efficient observed-data influence function}, which forms the basis of many estimators encountered in practice. In particular, given a full-data influence function $\finfluence{\data^F;\thetatarget}$, one can recover the efficient observed-data influence function using the following Lemma.
\begin{lemma}[Observed-data influence function]
\label{prop:obs-if}
Let $\thetatarget$ be a target parameter that admits a regular and asymptotically linear (RAL) expansion as in Eq.~\ref{eq:ral},
and let  $\finfluence{\data^F;\thetatarget}$ be the full-data influence function associated to it. Assume the identifiability conditions described in Assumption~\ref{ass:MAR}. Then, the observed-data influence function for the target $\thetatarget$ is given by:
\begin{equation}
    \label{eq:if}
        \influence{\data;\thetatarget} = \mutarget(X) + \frac{R}{\pitarget(X)} \left(\finfluence{\data;\thetatarget} - \mutarget(X) \right) \,,
\end{equation}
where $\mutarget(x) = \EE{\finfluence{\data^F;\thetatarget}\mid X=x, R=1}$ is the nuisance regression function, and $\pitarget : \RR^p \to (0,1)$ denotes the nuisance propensity score, defined as $\pitarget(x) = \PP{R = 1 \mid X = x}$.
\end{lemma}

In observational studies, the nuisance regression function $\mutarget$ and the propensity score $\pitarget$ are unknown and must be estimated from the data. For simplicity, we will proceed by assuming access to pre-trained models, denoted $\hat\mu$ and $\hat\pi$. However, our results can be readily extended to the practical setting where these nuisance functions are estimated from the data using techniques like \textit{sample splitting} or \textit{cross-fitting}. Sample splitting works as follows. We randomly split the observations $\{\data_1,\dots,\data_{n}\}$ into 2 disjoint folds. We form $\hat{\mathbbmss{P}}$ with the first fold, and $\mathbbmss{P}_{n}$ with the second fold. Then, we learn $\hat{\mu}$ and $\hat{\pi}$ on $\hat{\mathbbmss{P}}$, and we compute the estimator $\dr$ by solving for $\thetatarget$ the estimating equation
\begin{equation}
\sum_{i\in\mathbbmss{P}_{n}}\influence{\data_i;\thetatarget;\hat{\mu};\hat{\pi}} = 0\,.
\end{equation} 
This separation of training and estimation prevents overfitting and is crucial for valid inference. The resulting estimator takes the form
\begin{equation}
\label{eq:estimator}
    \dr = \frac{1}{\left|\mathbbmss{P}_{n}\right|} \sum_{i\in\mathbbmss{P}_{n}} \left[ \hat\mu(X_i) + \frac{R_i}{\hat\pi(X_i)} \left(\finfluence{\data_i;\thetatarget} - \hat\mu(X_i) \right) \right]\,,
\end{equation}
where $\left|\mathbbmss{P}_{n}\right|$ denotes the cardinality of $\mathbbmss{P}_{n}$. By assuming access to external nuisance models, the averaging distribution $\mathbbmss{P}_{n}$ contains the entire sample at hand, so that $\left|\mathbbmss{P}_{n}\right| = n$.

\subsection{What is robust in double robustness?}
The celebrated \textit{double robustness} property of estimators as in Eq.~\ref{eq:estimator} is foundational to modern semiparametric statistics. This Section revisits this concept, first through a novel framing that explains its mechanism and then by analyzing the failure mode that motivates our work.

\begin{figure}
    \centering
    \includegraphics[width=0.5\linewidth]{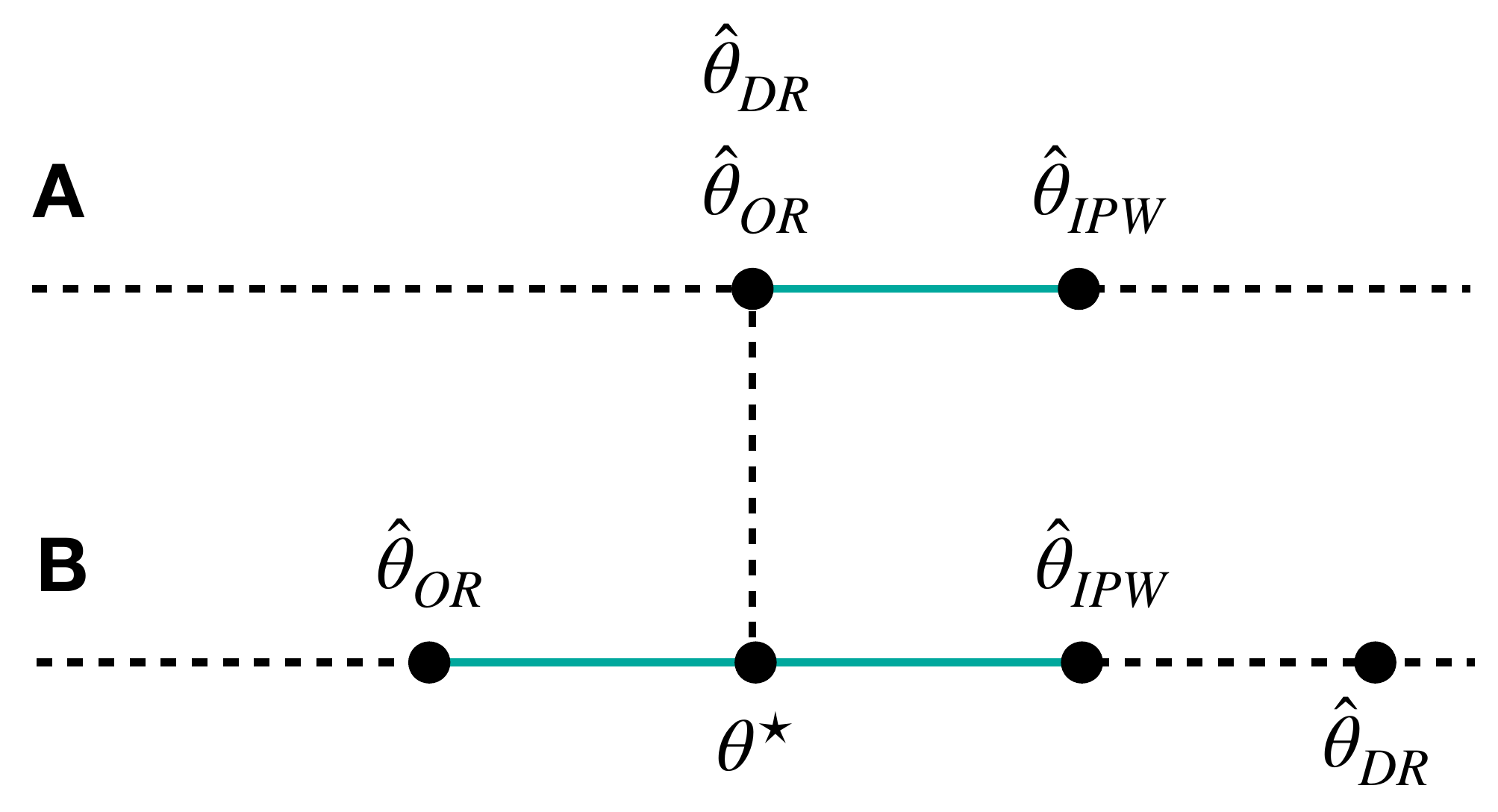}
    \caption{A conceptual illustration of the behavior of the doubly robust estimator under different model specification scenarios. Scenario A depicts the case of partial misspecification, where the \texttt{OR} estimator is consistent for $\thetatarget$ but the \texttt{IPW} estimator is biased. The \texttt{DR} estimator correctly aligns with the consistent \texttt{OR} estimator, demonstrating its asymptotic hard thresholding property. Scenario B depicts the case of complete misspecification, where both the \texttt{OR} and \texttt{IPW} estimators are biased. This illustrates double fragility: the bias of the \texttt{DR} estimator can be substantially worse than that of its constituent estimators, as its correction term can amplify, rather than reduce, the overall error.}
    \label{fig:sketch}
\end{figure}

Classically, double robustness refers to the property that an estimator remains consistent if at least one of its two nuisance functions (the outcome regression $\hat\mu$ or the propensity score $\hat\pi$) is correctly specified. We argue that this behavior can be more mechanistically understood as a form of \textit{asymptotic hard thresholding}. In fact, the \texttt{DR} estimator in Eq.~\ref{eq:estimator} can be written as:
\begin{equation}
    \dr = \oreg + \ipw - \frac{1}{\left|\mathbbmss{P}_{n}\right|} \sum_{i\in\mathbbmss{P}_{n}} \frac{R_i\hat\mu(X_i)}{\hat\pi(X_i)}\,.
\end{equation}
When one nuisance model is correct, the estimator correction term is engineered to asymptotically cancel the bias introduced by the misspecified model, effectively forcing the estimator to rely solely on the correctly specified component. When both models are correct, this flexibility allows the estimator to achieve optimal semiparametric efficiency. While Figure \ref{fig:sketch} depicts this intuition, the following Proposition formalizes it.
\begin{proposition}
\label{prop:corr}
    Assume the identifiability conditions described in Assumption~\ref{ass:MAR}. If $\hat\mu =\mutarget$, then
    \begin{equation}
        \EE{\frac{1}{n} \sum_{i=1}^n  \frac{R_i\hat\mu(X_i)}{\hat\pi(X_i)}} = \EE{\ipw}\,.
    \end{equation}
    Similarly, if $\hat\pi =\pitarget$, then
    \begin{equation}
        \EE{\frac{1}{n} \sum_{i=1}^n  \frac{R_i\hat\mu(X_i)}{\hat\pi(X_i)}} = \EE{\oreg}\,.
    \end{equation}
    In either case, $\EE{\dr} = \thetatarget$.
\end{proposition}

A more refined analysis of the error of the \texttt{DR} estimator can be carried out using \textit{Von Mises expansion}. In fact, assuming that there exist $\Bar{\mu}$ and $\Bar{\pi}$ such that $\influence{\data;\thetatarget;\hat{\mu};\hat{\pi}}\Lto{2} \influence{\data;\thetatarget;\Bar{\mu};\Bar{\pi}}$, we can write:
\begin{equation}
    \dr - \thetatarget = \underbrace{\left(\mathbbmss{P}_{n} - \jPtarget\right) \left[\influence{\data;\thetatarget;\Bar{\mu};\Bar\pi}\right]}_{\text{CLT term}} + \underbrace{\left(\mathbbmss{P}_{n} - \jPtarget\right) \left[ \influence{\data;\thetatarget;\hat\mu;\hat\pi} - \influence{\data;\thetatarget;\Bar{\mu};\Bar\pi} \right]}_{\text{Empirical process term}} + \underbrace{\eta\left((\hat\mu,\hat\pi),(\mutarget,\pitarget)\right)}_{\text{Remainder term}} \,.
\end{equation}
The first term is the sample average of a fixed function, and so Central Limit Theorem applies. The second term is an empirical process, and under the mild condition on the $\mathbbmss{L}_2$-convergence of influence functions above it can be shown to be $o_\mathbbmss{P}(n^{-1/2})$, e.g.~by Lemma 2 in \citet{kennedy2020sharp}. The third term, known as \textit{remainder} term, can be directly evaluated, and equals the product of the estimation errors of the two nuisance models:
\begin{equation}
        \left|\eta\left((\hat\mu,\hat\pi),(\mutarget,\pitarget)\right)\right| = \left| \EE{(\hat\mu - \mutarget )\left(1 - \frac{\pitarget}{\hat\pi} \right)} \right| \leq  \norm{\hat\mu-\mutarget} \norm{1 - \frac{\pitarget}{\hat\pi}} \,,
\end{equation}
where, defining for a random variable $Z$ the norm operator $\norm{Z} = \sqrt{\EE{Z^2}}$, the last inequality holds by Cauchy-Schwarz. This property, often called \textit{rate} double robustness, \textit{strong} double robustness, or \textit{product bias}, is highly advantageous when one model is easy to estimate well (e.g., at a parametric rate) while the other is not \citep{wager2024causal}. However, this same structure becomes a critical liability under complete model misspecification. When both nuisance functions are incorrect, their errors no longer cancel. Instead, they can compound through this product term, causing the error of the \texttt{DR} estimator to be even larger than that of the simpler \texttt{OR} or \texttt{IPW} estimators. We term this failure mode \textit{double fragility}.

\section{Safe estimation through adaptive correction clipping}
\label{sec:acc}
Having established that the correction term is the source of double fragility, we now introduce a solution that directly targets this mechanism. Our proposal builds on a key insight: any doubly robust estimator can be decomposed into its simpler components plus a correction term. Instead of letting the correction term vary freely, our proposed solution adaptively constrains the correction term to prevent it from destabilizing the final estimate. We define the \textit{adaptive correction clipping} (\ACC) estimator as
\begin{equation}
    \acc = \oreg + \ipw - \clip{\frac{1}{n} \sum_{i=1}^n \frac{R_i \hat\mu(X_i)}{\hat\pi(X_i)}}\,,
\end{equation}
where, given a properly chosen slack term $\delta_n\geq0$,
\begin{equation}
   \clip{\frac{1}{n} \sum_{i=1}^n \frac{R_i \hat\mu(X_i)}{\hat\pi(X_i)}} = \begin{cases}
        \min\{\oreg, \ipw \} - \delta_n & \text{if}\quad \frac{1}{n} \sum_{i=1}^n \frac{R_i \hat\mu(X_i)}{\hat\pi(X_i)} \leq \min\{\oreg, \ipw\} - \delta_n \\
        \max\{\oreg, \ipw\} + \delta_n &\text{if}\quad \frac{1}{n} \sum_{i=1}^n \frac{R_i \hat\mu(X_i)}{\hat\pi(X_i)} \geq \max\{\oreg, \ipw\} + \delta_n \\
        \frac{1}{n} \sum_{i=1}^n \frac{R_i \hat\mu(X_i)}{\hat\pi(X_i)} & \text{otherwise}\,.
    \end{cases}
\end{equation}

\begin{remark}
    A key advantage of our \ACC{} estimator, inherited by the standard doubly robust estimator, is its flexibility; any \enquote{black-box} machine learning model can be used to fit the nuisance models for the outcome regression and the propensity score. This allows for the use of state-of-the-art predictive tools best suited for the data at hand -- including contemporary deep learning models and even large language models for text (or textualized) covariates -- without altering the fundamental statistical properties of the final \ACC{} estimator.
\end{remark}


Now that we have introduced our \ACC{} estimator, we can analyze its statistical guarantees. First, we show that \ACC{} is consistent whenever at least one nuisance function is well-specified. Second, we show that the clipping mechanism becomes asymptotically inactive if the nuisance models converge to their population counterparts at a parametric product rate. This property ensures that \ACC{} is asymptotically equivalent to the standard \texttt{DR} estimator when the latter is optimal, thus fully preserving its \textit{semiparametric efficiency} and enabling valid inference based on asymptotic normality. Third, we show that the projection onto the clipping interval endows the \ACC{} estimator with two additional desirable properties: a \textit{safety} guarantee that provides protection under complete nuisance misspecification and an \textit{adaptive oracle} property on its mean squared error. We now make these claims more formal, starting with consistency.

\begin{theorem}[Consistency]
\label{th:consistency}
Assume the identifiability conditions described in Assumption~\ref{ass:MAR}. Assume also that at least one of the nuisance models, $\hat{\mu}$ or $\hat{\pi}$, is well-specified. Then, the \ACC{} estimator is consistent for $\thetatarget$, that is
\begin{equation}
    \acc \pto \thetatarget\,.
\end{equation}
\end{theorem}

\begin{remark}[Double robustness]
    This consistency result shows that our \ACC{} estimator is doubly robust, inheriting this critical property from the standard \texttt{DR} estimator. By remaining consistent when at least one nuisance model is correctly specified, the \ACC{} estimator offers a significant advantage over methods like \texttt{OR} and \texttt{IPW} estimators, which are only valid if their respective underlying models are correct. This makes our proposal a more reliable choice for practitioners, providing the same protection against partial misspecification as the classical \texttt{DR} approach.
\end{remark}

We now turn to the question of optimality. In the semiparametric framework, the efficiency bound provides the natural benchmark: it is the smallest asymptotic variance achievable by any RAL estimator, and the standard \texttt{DR} estimator is known to attain it whenever the product rate is parametric. A key desideratum for our \ACC{} proposal is therefore that it does not sacrifice this optimality in well-specified settings. The following result shows that no such trade-off occurs: when the nuisance models converge at a sufficiently fast product rate, the clipping mechanism becomes asymptotically inactive and \ACC{} inherits the full efficiency of \texttt{DR}. This equivalence enables valid inference based on asymptotic normality without any modification to standard procedures.

\begin{theorem}[Semiparametric efficiency]
\label{th:efficiency}
    Assume the identifiability conditions described in Assumption~\ref{ass:MAR}. Assume also that $\norm{\hat\mu - \mutarget} = o_\mathbbmss{P}\left(n^{-\alpha} \right)$ and $\norm{(\hat\pi - \pitarget)/\hat\pi} = o_\mathbbmss{P}\left(n^{-\beta} \right)$ with $\alpha + \beta = 1/2$. If the slack sequence $\delta_n$ satisfies $\delta_n n^\gamma \to \infty$, where $\gamma=\max\{\alpha, \beta\}$, then the \ACC{} estimator is asymptotically equivalent to the \texttt{DR} estimator:
    \begin{equation}
        \acc = \dr + o_\mathbbmss{P}\left( n^{-1/2} \right)\,.
    \end{equation}
\end{theorem}

\begin{remark}[Choice of $\delta_n$] The slack sequence $\delta_n$ governs a fundamental trade-off between safety and efficiency. Theorem~\ref{th:efficiency} requires $\delta_n n^\gamma \to \infty$ where $\gamma = \max\{\alpha, \beta\}$ is the faster of the two nuisance convergence rates. This ensures the clipping interval eventually contains the correction term whenever the nuisance models converge at the product rate, making the clipping asymptotically inactive and preserving semiparametric efficiency. In practice, $\gamma$ is unknown, but a conservative choice is $\gamma = 1/4$, which corresponds to the minimal rate required for standard doubly robust estimators to achieve $\sqrt{n}$-consistency. This yields the default recommendation $\delta_n = \sigma_Y \log(n)/n^\gamma$, where $\sigma_Y$ is a scale factor that can be estimated from the data. Larger values of $\delta_n$ widen the clipping interval, making clipping less likely and thus prioritizing efficiency at the expense of safety protection. Smaller values of $\delta_n$ tighten the interval, increasing the frequency of clipping and thus prioritizing safety at the expense of efficiency. Intuitively, $\delta_n$ should be thought of as the tolerance within which a well-behaved \texttt{DR} estimator is expected to fall: $\delta_n$ is chosen to generate a clipping interval ensuring that, with high probability, the \texttt{DR} estimator is not clipped if at least one nuisance is correctly specified. The special case $\delta_n = 0$, analyzed separately in Supplementary Section~\ref{sec:delta0}, provides the strongest safety guarantee -- the \ACC{} estimate is always a convex combination of \texttt{OR} and \texttt{IPW} -- but may clip even when \texttt{DR} is well-behaved, potentially sacrificing some efficiency. We discuss this case in detail in Supplementary Section~\ref{sec:delta0}.
\end{remark}

We now move to safety, which is the key statistical property of our estimator.

\begin{theorem}[Pointwise safety]
\label{th:safety}
    Assume the identifiability conditions described in Assumption~\ref{ass:MAR}. Define 
    \begin{equation}
        \lambda = \frac{\clip{\frac{1}{n}\sum_{i=1}^n \frac{R_i \hat\mu(X_i)}{\hat\pi(X_i)}} - \left( \min \left\{\oreg, \ipw \right\} -  \delta_n \right)}{\left(\max \left\{\oreg, \ipw \right\} + \delta_n\right) - \left(\min \left\{\oreg, \ipw \right\} - \delta_n\right)} \in [0,1]\,.
    \end{equation}
    Then, we have
    \begin{equation}
    \begin{split}
        \left|\acc - \thetatarget \right| &\leq \Tilde\lambda \left|\oreg - \thetatarget\right| + \left( 1 -\Tilde\lambda\right) \left|\ipw - \thetatarget\right| + \delta_n \\
        &\leq \max\left\{\left|\oreg - \thetatarget\right|, \left|\ipw - \thetatarget\right| \right\} + \delta_n  \,,
    \end{split}
    \end{equation}
    where $\Tilde{\lambda} = \lambda\onea{\oreg \leq \ipw} + (1-\lambda)\onea{\oreg > \ipw} \in \{0,1\}$.
\end{theorem}

\begin{remark}
    This safety property holds regardless of the quality of the nuisance function estimates. It provides a strict guarantee that the \ACC{} estimator error is bounded by the worst of its constituent parts, which implies that the clipped estimator is better than the standard \texttt{DR} estimator whenever the latter fails by performing worse than its components. This, combined with the consistency result in Theorem \ref{th:consistency} and Theorem \ref{th:efficiency}, creates a powerful set of assurances. If at least one nuisance model is well-specified, the estimator is consistent for the true parameter. If nuisance models converge to their population counterparts at a parametric product rate, then \ACC{} preserves the semiparametric efficiency of \texttt{DR}. If both models are misspecified, the estimator performance is guaranteed to be no worse than the maximum of its components. This behavior is fundamentally different from the standard \texttt{DR} estimator, which can be catastrophically wrong in the same scenario. The reason for this difference can be understood by comparing their errors. The error of the standard \texttt{DR} estimator depends on the product of the errors of the two nuisance models. In contrast, the error of the \ACC{} estimator is bounded by a convex combination of the errors of the \texttt{OR} and \texttt{IPW} estimators, as shown in Theorem \ref{th:safety}. In turn, this convex combination is bounded by the maximum of the two errors. When both nuisance models are substantially incorrect, the product of their errors can be far larger than their maximum. This is, again, the source of the double fragility phenomenon, where the standard \texttt{DR} estimator amplifies errors, while the \ACC{} contains them.
\end{remark}

The maximum bound provided in Theorem~\ref{th:safety} is a worst-case guarantee and is often conservative. As observed in simulations (e.g., Figure~\ref{fig:intro_plot_n1000} and Section~\ref{sec:sim} below), the performance of the \ACC{} estimator is often much better than this upper bound and can in fact be closer to the minimum of the errors of its component estimators. To understand why, it is helpful to revisit the geometry of the clipping mechanism. The \ACC{} estimator coincides with \texttt{DR} whenever the latter falls inside the clipping interval, and falls back to the base estimators only when \texttt{DR} strays outside it. Crucially, this distinction tracks model specification: $\PP{\dr \in \left[\min\{\oreg, \ipw\} - \delta_n, \max\{\oreg, \ipw\} + \delta_n\right]}\to 1$ if $\dr \pto \thetatarget$, which in turn holds if at least one nuisance model is correctly specified. In other words, clipping is active only if double fragility occurs. When at least one nuisance is correct, \ACC{} inherits the favorable behavior of \texttt{DR}; when both are misspecified, \ACC{} automatically switches to the protection offered by the base estimators. The safety bound captures the worst case of this switch, but a sharper picture emerges by characterizing the risk of \ACC{} relative to the best available estimator in each regime. This is the content of the following oracle inequality.

\begin{theorem}[Adaptive oracle property]
\label{th:oracle}
    Assume the identifiability conditions described in Assumption~\ref{ass:MAR}. Denote the squared loss of a generic estimator $\hat\theta$ as $\ell(\hat\theta) = (\hat\theta-\thetatarget)^2$ and the MSE risk as $\risk{\hat\theta} = \EE{\ell(\hat\theta)}$. Then the \ACC{} estimator satisfies the following bounds simultaneously:
    \begin{itemize}
        \item \textbf{\texttt{DR}-aware pointwise bound.} For every $j\in\{\texttt{OR}, \texttt{IPW}, \texttt{DR}\}$, 
        \begin{equation}
        \risk{\acc} \leq \risk{\hat\theta_j} + \Pi_j \,,
    \end{equation}
    where the penalty term $\Pi_j$ is defined as:
    \begin{equation}
        \Pi_j = \EE{\left(\ell^\star - \ell(\hat\theta_j)\right)_+}\,,
    \end{equation}
    with $\ell^\star = \ell(\dr) \onea{E} + \min\{\ell(\oreg),\ell(\ipw)\} \onea{E^C}$ being the loss of an oracle knowing when the event $E = \left\{ (\dr - \acc) (\thetatarget - \acc) \leq 0 \right\}$ is verified.
    
    \item \textbf{\texttt{DR}-free Minkowski bound.} For every $\eta>0$, 
    \begin{equation}
        \risk{\acc} \leq (1+\eta)\min\left\{\risk{\oreg}, \risk{\ipw}\right\} +
        \left(1+\frac{1}{\eta}\right) V_n\,,
    \end{equation}
    where the disagreement term $V_n$ satisfies:
    \begin{equation}
         V_n \leq 2 \left( \EE{\left(\oreg - \ipw \right)^2} + \delta_n^2\right)\,.
    \end{equation}
    \end{itemize}
\end{theorem}

\begin{remark}
Theorem~\ref{th:oracle} establishes two complementary risk bounds that together characterize the adaptive behavior of \ACC{}. The \texttt{DR}-aware bound compares \ACC{} to each candidate estimator via the event-wise oracle loss $\ell^\star$ -- the loss of an infeasible oracle that uses \texttt{DR} on $E$ and the better of \texttt{OR} and \texttt{IPW} on $E^C$. The \ACC{} estimator dominates this oracle pointwise, with a penalty that vanishes whenever clipping is asymptotically inactive. The \texttt{DR}-free Minkowski bound provides a complementary guarantee requiring no conditions on \texttt{DR}: the risk of \ACC{} is controlled by the best base estimator risk plus a penalty measuring disagreement between \texttt{OR} and \texttt{IPW}. This bound is immune to double fragility. The two bounds have complementary failure modes -- the first is sharp when \texttt{DR} is well-behaved, the second when it is not -- and their minimum automatically selects the tighter guarantee depending on the regime, without \ACC{} requiring to know which applies.
\end{remark}

\section{Simulation study}
\label{sec:sim}
To empirically validate the concept of double fragility and the safety of our proposed \ACC{} estimator, we conduct a simulation study that fully replicates the well-known design of \citet{kang2007demystifying}. This setup is ideal for our purposes as it is explicitly designed to create the possible scenarios of nuisance model specification, allowing for a rigorous evaluation of estimator performance. That is, the simulation is designed to create scenarios where both the models for the regression function and the propensity score can be specified either correctly or incorrectly. For each unit $i = 1, \dots, n$ we generate data through the following steps:
\begin{itemize}
    \item Latent variable generation: 4 independent latent variables are drawn from a standard normal distribution:
    \begin{equation}
        (T_{i1}, T_{i2}, T_{i3}, T_{i4}) \sim\Normal{0}{I_4}\,,
    \end{equation}
    where $I_4$ is the $4 \times 4$ identity matrix. These latent variables form the basis for both the outcome and the missingness mechanism.
    \item Outcome generation: the outcome variable, $Y_i$, is generated as a linear function of the latent variables $T_{ij}$ with an added standard normal error term, $\varepsilon_i$:
    \begin{equation}
        Y_i = 210 + 27.4T_{i1} + 13.7T_{i2} + 13.7T_{i3} + 13.7T_{i4} + \varepsilon_i\,, \quad \text{where} \quad \varepsilon_i \sim \Normal{0}{1}\,.
    \end{equation}
    This constitutes the true outcome model. The resulting population mean is $\thetatarget = \EE{Y_i} = 210$.
    \item Missingness mechanism: the probability of the outcome $Y_i$ being observed, denoted $\pi(T_i)$, is determined by a logistic regression model that is also linear in the latent variables:
    \begin{equation}
        \text{logit}(\pi(T_i)) = \text{logit}\left(\PP{R_i=1 \mid T_{i1},T_{i2},T_{i3},T_{i4}}\right) = -T_{i1} + 0.5T_{i2} - 0.25T_{i3} - 0.1T_{i4}\,.
    \end{equation}
    This defines the true propensity score model. The response indicator for each unit, $R_i$, is then drawn from a Bernoulli distribution with this probability:
    \begin{equation}
        R_i \sim \text{Bernoulli}(\pi(T_i))\,.
    \end{equation}
    \item Observed covariates: instead of observing the true latent variables, the analyst is presented with a set of four covariates which are non-linear transformations of the original ones:
    \begin{equation}
            X_{i1} = \exp\left(\frac{T_{i1}}{2}\right)\,, \quad
            X_{i2} = \frac{T_{i2}}{1+\exp(T_{i1})} + 10\,,\quad
            X_{i3} = \left(\frac{T_{i1}T_{i3}}{25} + 0.6\right)^3\,, \quad
            X_{i4} = (T_{i2} + T_{i4} + 20)^2\,.
    \end{equation}
    An analyst who fits a linear model for the outcome or a logistic model for the propensity score using these observed covariates $X_i$ will have a misspecified model. A correct specification would require the analyst to know the true latent variables $T_i$ or the exact inverse transformations.
\end{itemize}

\begin{figure}[t]
    \centering
    \includegraphics[width=\linewidth]{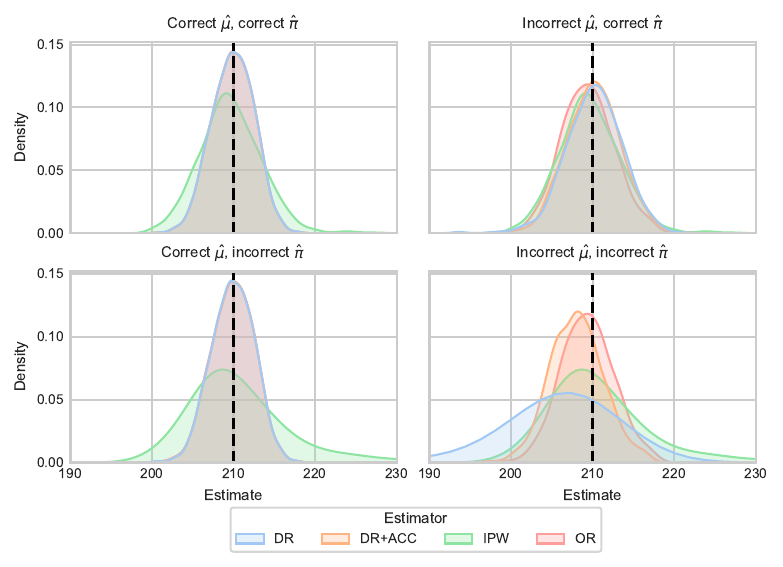}
    \caption{Sampling distributions for the Outcome Regression (\texttt{OR}), Inverse Probability Weighting (\texttt{IPW}), Doubly Robust (\texttt{DR}), and our proposed \ACC{} estimators from 1000 simulations with a sample size of $n=200$. The true parameter value is 210. The four panels show the estimators' performance under all combinations of correct and incorrect nuisance model specifications. The top row and bottom-left panel demonstrate the \textit{asymptotic hard thresholding} property. When at least one nuisance model is correct, both \texttt{DR} and \ACC{} align with the correct simpler estimator. The bottom-right panel illustrates \textit{double fragility}. When both nuisance models are wrong, the bias of the standard \texttt{DR} estimator is substantially worse than that of either the \texttt{OR} or \texttt{IPW} estimators. In this challenging scenario, our proposed \ACC{} estimator is shown to be \textit{safe}, providing a much more stable and accurate estimate than the standard \texttt{DR} estimator. Full details of the simulation scenario are provided in Section \ref{sec:sim}.}
    \label{fig:intro_plot_n200}
\end{figure}

\begin{table}
\centering
\caption{Simulation results for $n=200$ and $n=1000$, grouped by nuisance model specification. The results highlight the double fragility of the standard \texttt{DR} estimator and the safety of the \ACC{} estimator across both sample sizes.}
\label{tab:sim_results}
\begin{tabular}{ll rrr rrr}
\toprule
& & \multicolumn{3}{c}{$n = 200$} & \multicolumn{3}{c}{$n = 1000$} \\
\cmidrule(lr){3-5} \cmidrule(lr){6-8}
\textbf{Scenario} & \textbf{Estimator} & \textbf{Bias} & \textbf{RMSE} & \textbf{MAE} & \textbf{Bias} & \textbf{RMSE} & \textbf{MAE} \\
\midrule
\multirow{4}{*}{Correct $\hat\mu$, Correct $\hat\pi$} & \texttt{OR} & -0.072 & 2.568 & 1.853 & -0.002 & 1.128 & 0.761 \\
 & \texttt{IPW} & -0.311 & 3.859 & 2.464 & -0.019 & 1.688 & 1.098 \\
 & \texttt{DR} & -0.072 & 2.570 & 1.851 & -0.002 & 1.128 & 0.767 \\
 & \ACC{} & -0.072 & 2.570 & 1.851 & -0.002 & 1.128 & 0.767 \\
\midrule
\multirow{4}{*}{Correct $\hat\mu$, Incorrect $\hat\pi$} & \texttt{OR} & -0.072 & 2.568 & 1.853 & -0.002 & 1.128 & 0.761 \\
 & \texttt{IPW} & 1.595 & 9.726 & 3.412 & 4.761 & 11.095 & 2.561 \\
 & \texttt{DR} & -0.076 & 2.574 & 1.840 & 0.018 & 1.314 & 0.786 \\
 & \ACC{} & -0.073 & 2.569 & 1.840 & 0.025 & 1.302 & 0.784 \\
\midrule
\multirow{4}{*}{Incorrect $\hat\mu$, Correct $\hat\pi$} & \texttt{OR} & -0.665 & 3.306 & 2.273 & -0.814 & 1.678 & 1.179 \\
 & \texttt{IPW} & -0.311 & 3.859 & 2.464 & -0.019 & 1.688 & 1.098 \\
 & \texttt{DR} & 0.160 & 3.448 & 2.268 & 0.028 & 1.643 & 1.069 \\
 & \ACC{} & 0.082 & 3.242 & 2.194 & -0.148 & 1.524 & 1.054 \\
\midrule
\multirow{4}{*}{Incorrect $\hat\mu$, Incorrect $\hat\pi$} & \texttt{OR} & -0.665 & 3.306 & 2.273 & -0.814 & 1.678 & 1.179 \\
 & \texttt{IPW} & 1.595 & 9.726 & 3.412 & 4.761 & 11.095 & 2.561 \\
 & \texttt{DR} & -6.395 & 21.932 & 4.026 & -13.467 & 77.565 & 5.295 \\
 & \ACC{} & -1.962 & 3.812 & 2.620 & -1.601 & 2.160 & 1.663 \\
\bottomrule
\end{tabular}
\end{table}

We first assess performance in terms of estimation accuracy, which we measure through three different metrics. \textit{Bias} measures the difference between the estimated mean and the true mean ($\thetatarget=210$). The \textit{root mean squared error} (RMSE) is the square root of the average squared difference between the estimate and the true mean, providing a comprehensive measure of an estimator's overall accuracy by combining both its bias and variance. Finally, the \textit{median absolute error} (MAE) is a robust measure of precision, calculated as the median of the absolute errors. 

Figures \ref{fig:intro_plot_n1000} and \ref{fig:intro_plot_n200}, together with Table \ref{tab:sim_results}, summarize the results of our simulation study, providing strong empirical support for our central arguments across both sample sizes ($n=200$ and $n=1000$).
In the three scenarios where at least one nuisance model is correctly specified, the results align with classical theory. Both the standard doubly robust (\texttt{DR}) and our proposed \ACC{} estimators perform excellently, exhibiting minimal bias and RMSE that is comparable to the best-performing correctly specified estimator (\texttt{OR} or \texttt{IPW}). This confirms their double robustness and the asymptotic hard thresholding property in practice. The most critical insights come from the scenario where both nuisance models are misspecified. Here, the \texttt{DR} estimator fails catastrophically, demonstrating the phenomenon of double fragility. As shown in Table \ref{tab:sim_results}, its RMSE explodes, increasing from approximately 2.570 to 21.932 for $n=200$, and from approximately 1.128 to 77.565 for $n=1000$ -- an error far exceeding that of either the misspecified \texttt{OR} or \texttt{IPW} estimators. In stark contrast, our \ACC{} estimator remains stable in this challenging setting. Its RMSE is dramatically lower than that of the standard \texttt{DR} estimator (3.812 vs.~21.932 for $n=200$; 2.160 vs.~77.565 for $n=1000$) and is comparable to the better between the two simpler \texttt{OR} and \texttt{IPW} estimators. Similarly, Supplementary Figures \ref{fig:scatter200} and \ref{fig:scatter1000} display the relationship between the estimates provided by \ACC{} and the other estimators. Again, \ACC{} estimates follow the standard \texttt{DR} when the latter is stable but remains bounded when the standard \texttt{DR} produces extreme, unstable estimates. These results (together with the ones for $n=100$ in Supplementary Figures \ref{fig:intro_plot_n100} and \ref{fig:scatter100}) empirically validate the safety property of our method, demonstrating that it successfully mitigates the critical failure mode of the standard \texttt{DR} approach without sacrificing performance in well-behaved settings.

We also evaluate performance in terms of inference, using two metrics: \textit{empirical coverage} (the fraction of experiments in which confidence intervals contain the true parameter) and the average confidence interval (CI) \textit{width}. The results for all four nuisance specification scenarios are presented in Table \ref{tab:coverage_results}. The results in the ideal scenario where both the outcome regression and propensity score models are correctly specified, the only scenario in which confidence intervals are theretically guaranteed to be asymptotically valid \citep{kennedy2024semiparametric}, validate that the inference procedure based on asymptotic normality for the \ACC{} estimator achieves the nominal 95\% coverage level, performing almost identically to the standard \texttt{DR} and \texttt{OR} estimators at both sample sizes ($n=200$ and $n=1000$). In terms of efficiency, the \texttt{IPW} estimator is clearly suboptimal, producing extremely wide confidence intervals. In contrast, the \texttt{DR} and \ACC{} estimators are highly efficient, yielding CIs that are nearly identical in width to the efficient \texttt{OR} estimator. This confirms that the safety mechanism of the \ACC{} estimator does not compromise its statistical efficiency in well-behaved settings where the risk of fragility is absent. While our theoretical guarantees do not cover misspecified settings, our simulation results also provide valuable insights into the practical robustness of our approach. When one of the two nuisance functions is misspecified, our \ACC{} confidence intervals remain bonded to the standard \texttt{DR} CIs. In addition, when both nuisance functions are wrong, the advantage of our approach becomes even clearer. The standard \texttt{DR} confidence intervals suffer from undercoverage, dropping to 72.2\% at $n=1000$. In contrast, our \ACC{} confidence intervals maintain nominal coverage (96.6\%). This, together with the results for $n=100$ in Supplementary Table \ref{tab:coverage_results_n100}, suggests that the safety property of the \ACC{} estimator not only improves point estimation but also leads to more reliable inference in the realistic setting of complete model misspecification.

\begin{table}
\centering
\caption{Empirical coverage and confidence interval (CI) widths for $n=200$ and $n=1000$, grouped by nuisance model specification. The nominal coverage level is 95\%.}
\label{tab:coverage_results}
\begin{tabular}{ll rr rr}
\toprule
& & \multicolumn{2}{c}{$n = 200$} & \multicolumn{2}{c}{$n = 1000$} \\
\cmidrule(lr){3-4} \cmidrule(lr){5-6}
\textbf{Scenario} & \textbf{Estimator} & \textbf{Coverage} & \textbf{Width} & \textbf{Coverage} & \textbf{Width} \\
\midrule
\multirow{4}{*}{Correct $\hat\mu$, Correct $\hat\pi$} & \texttt{OR} & 0.949 & 10.00 & 0.952 & 4.48 \\
 & \texttt{IPW} & 1.000 & 92.51 & 1.000 & 41.49 \\
 & \texttt{DR} & 0.950 & 10.01 & 0.951 & 4.49 \\
 & \ACC{} & 0.950 & 10.01 & 0.951 & 4.49 \\
\midrule
\multirow{4}{*}{Correct $\hat\mu$, Incorrect $\hat\pi$} & \texttt{OR} & 0.949 & 10.00 & 0.952 & 4.48 \\
 & \texttt{IPW} & 1.000 & 164.35 & 1.000 & 150.79 \\
 & \texttt{DR} & 0.952 & 10.03 & 0.955 & 4.63 \\
 & \ACC{} & 0.952 & 10.03 & 0.955 & 4.63 \\
\midrule
\multirow{4}{*}{Incorrect $\hat\mu$, Correct $\hat\pi$} & \texttt{OR} & 0.875 & 10.30 & 0.821 & 4.58 \\
 & \texttt{IPW} & 1.000 & 92.51 & 1.000 & 41.49 \\
 & \texttt{DR} & 0.963 & 13.68 & 0.970 & 6.63 \\
 & \ACC{} & 0.963 & 13.68 & 0.973 & 6.63 \\
\midrule
\multirow{4}{*}{Incorrect $\hat\mu$, Incorrect $\hat\pi$} & \texttt{OR} & 0.875 & 10.30 & 0.821 & 4.58 \\
 & \texttt{IPW} & 1.000 & 164.35 & 1.000 & 150.79 \\
 & \texttt{DR} & 0.922 & 27.58 & 0.722 & 43.50 \\
 & \ACC{} & 0.952 & 27.58 & 0.966 & 43.50 \\
\bottomrule
\end{tabular}
\end{table}

\section{Application to Alzheimer's disease proteomics}
\label{sec:app}
To demonstrate the practical efficacy of our method in a real-world scientific setting, we revisit an application from \citet{moon2025augmented}, focusing on the impact of Alzheimer's disease (AD) on the human proteome. AD is a prominent neurodegenerative disorder, and while many contributing factors are known, its underlying biological pathways are still being discovered. Bulk peptide-level datasets offer a valuable opportunity to explore these pathways, but they also present complex modeling challenges where the risk of misspecification is high. Our objective is to estimate the average treatment effect (ATE) of an AD diagnosis on the abundances of various peptides. This estimand, which is the difference between two means, fits naturally within our missing data framework.

We use the peptide abundance data from \citet{merrihew2023peptide}, which contains samples from individuals with and without dementia. Following previous work \citep{moon2025augmented}, we define our treatment variable by grouping samples into two categories: cases (individuals with autosomal dominant or sporadic AD dementia) and controls (individuals without dementia, with or without a high AD histopathologic burden). To ensure data quality, we focus on the 270 peptides that present no more than 10\% missing values across the 220 observations in the data set, and we impute the remaining missing values through \texttt{MissForest} \citep{stekhoven2012missforest}.

For our analysis, we estimate the two nuisance functions using standard approaches. For the outcome model, we use a difference-in-means estimator. While simple, this approach is routinely used in the analysis of peptide abundance data \citep{chen2020bioinformatics} and serves as a relevant baseline. For the propensity score model, we fit a logistic regression of the AD diagnosis on some available external covariates, namely brain region, post-mortem interval (PMI), age, and gender.
To improve stability, the estimated propensity scores are rescaled in the spirit of a Hájek estimator \citep{basu1971essay}: we first multiply each estimated propensity score $\Tilde{\pi}(X_i)$ by $n^{-1} \sum_{i=1}^n A_i/\Tilde{\pi}(X_i)$ and $n^{-1} \sum_{i=1}^n (1-A_i)/(1-\Tilde{\pi}(X_i))$ for treatment and control observations respectively, and then we clip the result so that it ranges between 0 and 1. We then compute the \texttt{OR}, \texttt{IPW}, \texttt{DR}, and our proposed \ACC{} estimators. 

The results, illustrated in Figure \ref{fig:app} and Supplementary Figure \ref{fig:app_scatter}, highlight the practical impact of double fragility in a real-world setting. For several peptides, the standard \texttt{DR} estimator produces ATE estimates that deviate from the \texttt{OR} and \texttt{IPW} estimates. In contrast, our \ACC{} estimator consistently provides more credible results that remain anchored between the two simpler \texttt{OR} and \texttt{IPW} estimators, demonstrating the importance of its safety property for drawing reliable scientific conclusions when the true data-generating process is unknown. 

\begin{figure}
    \centering
    \includegraphics[width=\linewidth]{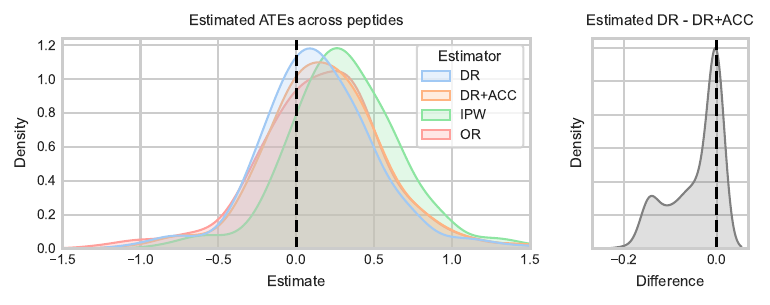}
    \caption{In the left panel, we show the distributions of the estimated average treatment effects (ATEs) of Alzheimer's disease across 270 different peptides. All estimators produce distributions centered near zero, suggesting that for most peptides, the estimated effect of AD is small. However, some of the estimates from the standard \texttt{DR} estimator are outside the range defined by the estimates of \texttt{OR} and \texttt{IPW}, providing a real-world example of double fragility. The stability of the \ACC{} estimator, which remains aligned with the more plausible \texttt{OR} and \texttt{IPW} results, demonstrates the practical importance of its safety property for drawing reliable scientific conclusions. In the right panel, we focus on the difference between the \texttt{DR} estimator and \ACC{}, plotting the distribution between the estimates provided by the two estimators.}
    \label{fig:app}
\end{figure}

We 
now present a preliminary analysis of the 
above estimates.
Switching from the standard \texttt{DR} 
to the proposed \ACC{} 
has substantial consequences.
Inference based on the standard \texttt{DR} estimator identifies 55 peptides as differentially abundant in cases vs.~controls (at a significance level of $\alpha=0.05$, without multiple testing correction). In contrast, inference using
the \ACC{} estimator
identifies 67 significant peptides. This includes 54 of the 55 peptides identified by the standard \texttt{DR} estimator, plus 13 additional findings unique to the \ACC{} estimator. These additional peptides map to proteins translated from 13 distinct primary genes. Notably, these 13 genes have been independently implicated in Alzheimer's disease by prior research, providing some external validation for our findings. A full list of the genes and the supporting literature is available in Supplementary Table \ref{tab:gene_list}. While a full biological interpretation is beyond the scope of this paper, these results suggest that the \ACC{} estimator, by mitigating the instability inherent in the standard \texttt{DR} approach, can indeed uncover credible and scientifically important biological signals.

\section{Conclusions}
\label{sec:end}

This paper identifies and formally characterizes double fragility, a critical failure mode of standard doubly robust estimators that arises under the ubiquitous scenario of complete nuisance model misspecification. While these estimators are celebrated for their theoretical properties, we show that the very mechanism that provides robustness -- which we call asymptotic hard thresholding -- can, in practice, amplify errors and lead to 
instability. 

To address this, we propose the \ACC{} estimator. Our solution is both simple and powerful.
When at least one nuisance model is correct, our \ACC{} estimator
retains the desirable double robustness property of standard methods, ensuring consistency. When the nuisance functions converge to their population counterparts at a parametric product rate, \ACC{} is also semiparametric efficient. 
At the same time, our \ACC{}
posesses a crucial safety guarantee, 
which ensures that its performance 
cannot be worse than that of its simpler constituent parts when all models are misspecified, thus preventing catastrophic failure.
Our simulation study and application to peptide abundance data confirm that the \ACC{} estimator provides a practical and reliable alternative, offering practitioners the benefits of double robustness with a crucial safeguard against its fragility.

This work opens several avenues for future research. First, our \ACC{} framework could be extended to handle multivariate or other complex target parameters. Second, 
while a key practical advantage of our 
\ACC{} estimator is 
that it does not require the tuning of any additional parameters, 
alternative 
solutions for achieving safety, such as power-tuning or data-adaptive convex combinations of the \texttt{OR} and \texttt{IPW} estimators \citep{angelopoulos2023ppi++}, could and should be explored. 
Finally, our simulation study suggests that the bias of 
doubly robust estimators can be exacerbated, not reduced, as the sample size grows. Exploring the relationship between sample size and bias of \texttt{DR} under complete misspecification represents a further interesting research direction.

\section*{Acknowledgments}
L.T.~wishes to thank the members of the causal inference reading group at Carnegie Mellon University, Ana M.~Kenney and the other participants to the 2025 L'EMbeDS workshop \enquote{Learning from large, complex and structured data} held in Pisa, Italy, in June 2025 for helpful discussions. This project was supported by National Institute of Mental Health (NIMH) grant R01MH123184.

\bibliographystyle{plainnat}
\bibliography{bib}

\clearpage
\setcounter{page}{1}
\appendix
\section*{Supplementary Material}

\renewcommand\thefigure{\thesection.\arabic{figure}}
\renewcommand\thetable{\thesection.\arabic{table}}
\setcounter{figure}{0} 
\setcounter{table}{0}

\section{Technical lemmas}
\begin{lemma}[Continuity of the clip operator]
\label{lemma:cont_clip}
The function $\clip{x, y, z} = \max(y, \min(x, z))$ is a continuous function from $\RR^3$ to $\RR$.
\end{lemma}

\begin{proof}
The functions $g_1(x, y, z) = x$, $g_2(x, y, z) = y$, and $g_3(x, y, z) = z$ are continuous. The binary operators $\min(a, b)$ and $\max(a, b)$ are continuous functions from $\RR^2 \to \RR$. The composition of continuous functions is continuous. Since $\clip{x, y, z}$ is a composition of these continuous functions, it is itself continuous.
\end{proof}

\begin{lemma}[Exchange of limit and clip]
\label{lemma:lim_clip}
Let $\{f_n\}_{n=1}^{\infty}$, $\{L_n\}_{n=1}^{\infty}$, and $\{U_n\}_{n=1}^{\infty}$ be sequences of real numbers that converge to the limits $f$, $L$, and $U$ respectively, that is,
\begin{equation*}
    \lim_{n \to \infty} f_n = f\,, \quad \lim_{n \to \infty} L_n = L\,, \quad \lim_{n \to \infty} U_n = U\,.
\end{equation*}
Then, the limit of the clipped sequence is the clip of the limits:
\begin{equation}
     \lim_{n \to \infty} \clip{f_n, L_n, U_n} = \clip{f, L, U} 
\end{equation}
\end{lemma}

\begin{proof}
Let the vector sequence be defined as $V_n = (f_n, L_n, U_n) \in \RR^3$. By the definition of convergence of a vector sequence, since each component converges, the vector sequence converges to the vector limit $V = (f, L, U)$, that is,
\begin{equation*}
    \lim_{n \to \infty} V_n = V\,.
\end{equation*}
By definition of function continuity for sequences, if a function $g: \RR^3 \to \RR$ is continuous at a point $V$, and a sequence $V_n \to V$, then the sequence $g(V_n)$ must converge to $g(V)$, that is,
\begin{equation*}
    \lim_{n \to \infty} g(V_n) = g(V)\,.
\end{equation*}
From Lemma~\ref{lemma:cont_clip}, we know that the function $g(V) = \clip{V}$ is continuous everywhere. We can therefore substitute our sequence $V_n$ and its limit $V$ into the continuity definition, getting
\begin{equation*}
    \lim_{n \to \infty} \clip{f_n, L_n, U_n} = \clip{f, L, U}\,,
\end{equation*}
which is equivalent to:
\begin{equation*}
    \lim_{n \to \infty} \clip{f_n, L_n, U_n} = \clip{\lim_{n \to \infty} f_n, \lim_{n \to \infty} L_n, \lim_{n \to \infty} U_n}\,.
\end{equation*}
\end{proof}

\begin{lemma}[Pointwise master inequality]
\label{lem:pointwise}
Pointwise on the sample space,
\begin{equation}
    \ell(\acc) \leq \ell^\star = \ell(\dr) \onea{E} + \min\{\ell(\oreg),\ell(\ipw)\} \onea{E^C}\,.
\end{equation}
\end{lemma}
\begin{proof}
Denote $L_n = \min\{\oreg,\ipw\} -\delta_n$, $U_n = \max\{\oreg,\ipw\} +\delta_n$, and $I_n = [L_n, U_n ]$. Partition the sample space into four mutually exclusive and exhaustive events:
    \begin{equation}
        \begin{split}
            E_{in} &= \{\thetatarget \in I_n\}, \\
            E_{out,mid} &= \{\thetatarget \notin I_n, \dr \in I_n\}, \\
            E_{out,opp} &= \{\thetatarget < L_n, \dr > U_n\} \cup \{\thetatarget > U_n, \dr < L_n\}, \\
            E_{out,same} &= \{\thetatarget < L_n, \dr < L_n\} \cup \{\thetatarget > U_n, \dr > U_n\}\,.
        \end{split}
    \end{equation}
    Aggregate the events in $E = E_{in} \cup E_{out,mid} \cup E_{out,opp}$ and $E^C = E_{out,same}$. 

    On $E_{in}$, the projection onto a closed convex set is non-expansive relative to points in the set, so $|\acc - \thetatarget| \le |\dr - \thetatarget|$, giving $\ell(\acc) \le \ell(\dr)$. 
    
    On $E_{out,mid}$, $\dr \in I_n$, so $\acc = \dr$ and $\ell(\acc) = \ell(\dr)$. 
    
    On $E_{out,opp}$, say $\thetatarget < L_n$ and $\dr > U_n$ (the symmetric case is identical), then $\acc = U_n$ and $\thetatarget < L_n \le U_n < \dr$, so $\acc$ lies strictly between $\thetatarget$ and $\dr$, yielding $\ell(\acc) \le \ell(\dr)$.
    
    On $E_{out,same}$, say $\thetatarget < L_n$ and $\dr < L_n$ (the symmetric case is again identical), then $\acc = L_n = \min\{\oreg,\ipw\} - \delta_n$. Since $\thetatarget < \min\{\oreg,\ipw\} - \delta_n < \min\{\oreg,\ipw\} \le \max\{\oreg,\ipw\}$, the nearer base estimator is $\min\{\oreg,\ipw\}$, so the best base error is $\min\{\oreg,\ipw\} - \thetatarget$. Then
    \begin{equation}
        |\acc - \thetatarget| = |\min\{\oreg,\ipw\} - \delta_n - \thetatarget| \leq \min\{|\oreg-\thetatarget|,|\ipw-\thetatarget|\}\,,
    \end{equation}
    hence $\ell(\acc) \leq \min\{\ell(\oreg), \ell(\ipw)\}$. Combining the four cases gives the stated bound.
\end{proof}

\begin{lemma}[Pointwise base-anchored bound]
\label{lem:base-Dn}
For each $j \in \{\texttt{OR}, \texttt{IPW}\}$, pointwise on the sample space,
\begin{equation}
    |\acc - \hat\theta_j| \leq |\oreg - \ipw| + \delta_n\,,
\end{equation}
and consequently
\begin{equation}
    |\acc - \thetatarget| \leq |\hat\theta_j - \thetatarget| + |\oreg - \ipw| + \delta_n\,.
\end{equation}
\end{lemma}

\begin{proof}
Since $\acc \in I_n = \left[ \min\{\oreg,\ipw\} - \delta_n, \max\{\oreg,\ipw\} + \delta_n \right]$ and $\hat\theta_j \in \left[\min\{\oreg,\ipw\}, \max\{\oreg,\ipw\} \right]$,
\begin{equation}
    |\acc - \hat\theta_j| \leq  \left( \max\{\oreg,\ipw\} - \min\{\oreg,\ipw\} \right) + \delta_n = |\oreg - \ipw| + \delta_n\,.
\end{equation}
The second inequality is the triangle inequality applied to
$|\acc - \thetatarget| \leq |\hat\theta_j - \thetatarget| + |\acc - \hat\theta_j|$.
\end{proof}

\begin{lemma}[Asymptotic expansion of clip] 
 \label{lem:aec}
 Let $\hat C = n^{-1} \sum_{i=1}^n R_i \hat\mu(X_i) / \hat\pi(X_i)$. Assume the asymptotic expansions in Equation \ref{eq:asymp_exp} hold. Then we have
    \begin{equation}
        \clip{\hat C, L_n, U_n} = \thetatarget + \frac{1}{\sqrt{n}}\clip{\zcorr, \min\{\zoreg, \zipw\}, \max\{\zoreg, \zipw\}} + o_p(n^{-1/2})\,,
    \end{equation}
    where $L_n = \min\{\oreg,\ipw\}$ and $U_n = \max\{\oreg,\ipw\}$.
\end{lemma}
\begin{proof}
The proof proceeds by substituting the given asymptotic expansions into the $\text{clip}$ function and simplifying, leveraging the properties of the $\min$, $\max$, and $\text{clip}$ operators. We first expand the clipping bounds. For $L_n$:
\begin{equation}
    L_n = \min\{\oreg, \ipw\} = \min\left\{\thetatarget + \frac{\zoreg}{\sqrt{n}} + o_p(n^{-1/2}), \thetatarget + \frac{\zipw}{\sqrt{n}} + o_p(n^{-1/2})\right\}\,.
\end{equation}
Since $\min(a+x, a+y) = a + \min(x,y)$, we can factor out the common $\thetatarget$ term, getting:
\begin{equation}
    L_n = \thetatarget + \min\left\{\frac{\zoreg}{\sqrt{n}} + o_p(n^{-1/2}),\frac{\zipw}{\sqrt{n}} + o_p(n^{-1/2})\right\}\,.
\end{equation}
The same logic applies to the upper bound $U_n$. Now we can substitute the expansions for $\hat C$, $L_n$, and $U_n$ into the main expression:
\begin{equation}
    \clip{\hat C, L_n, U_n} = \clip{\thetatarget + \frac{\zcorr}{\sqrt{n}} + o_p(n^{-1/2}), \thetatarget + \min\{\dots\}, \thetatarget + \max\{\dots\}}\,.
\end{equation}
The clip function is translation equivariant ($\clip{x+a, y+a, z+a} = a + \clip{x, y, z}$), so that we can factor out the common $\thetatarget$ term:
\begin{equation}
    \clip{\hat C, L_n, U_n} = \thetatarget + \clip{\frac{\zcorr}{\sqrt{n}} + o_p(n^{-1/2}), \min\left\{\frac{\zoreg}{\sqrt{n}} + \dots\right\}, \max\left\{\frac{\zipw}{\sqrt{n}} +\dots\right\}}
\end{equation}
The clip, min, and max functions are also scale equivariant for positive constants. We can factor out the $1/\sqrt{n}$ term:
\begin{equation}
    \clip{\hat C, L_n, U_n} = \thetatarget + \frac{1}{\sqrt{n}} \clip{ \zcorr + o_p(1), \min\{\zoreg + o_p(1), \zipw + o_p(1)\}, \max\{\dots\} }
\end{equation}
Because the clip, min, and max functions are all continuous, the small $o_p(1)$ terms (which converge to zero in probability) can be pulled out of the functions, resulting in a single remainder term:
\begin{equation}
    \clip{\hat C, L_n, U_n} = \thetatarget + \frac{1}{\sqrt{n}}\clip{\zcorr, \min\{\zoreg, \zipw\}, \max\{\zoreg, \zipw\}} + o_p(n^{-1/2})\,.
\end{equation}
This completes the proof.
\end{proof}

\section{Proof of main results}

\subsection{Proof of Lemma~\ref{prop:obs-if}}
\begin{proof}
We write $\shortfif$ to indicate the full-data influence function for notational simplicity. 
Given a full-data influence function $\shortfif$, by Theorem 7.2 in \citet{tsiatis2006semiparametric}, we know that the space of associated observed-data influence functions $\Lambda^\perp$ is given by 
\begin{equation}
\label{eq:class_oif}
    \Lambda^\perp = \left\{\frac{R}{\pitarget(X)} \shortfif \oplus \Lambda_2 \right\}\,,
\end{equation}
where $ \Lambda_2 = \left\{ L_2(\data) \,:\, \EE{L_2(\data)\mid \data^F} = 0 \right\}$. 

By Theorem 10.1 in \citet{tsiatis2006semiparametric}, for a fixed $\shortfif$, the optimal observed-data influence function among the class in Equation~\ref{eq:class_oif} is obtained by choosing
\begin{equation}
    L_2(\data) = - \Pi\left(\frac{R}{\pitarget(X)} \shortfif \mid \Lambda_2 \right)\,,
\end{equation}
where the operator $\Pi$ projects the element $R\shortfif/\pitarget(X)$ onto the space $\Lambda_2$ -- see Theorem 2.1 in \citet{tsiatis2006semiparametric}. 

Finally, by Theorem 10.2 in \citet{tsiatis2006semiparametric}, we know that 
\begin{equation}
    \Pi\left(\frac{R}{\pitarget(X)} \shortfif \mid \Lambda_2 \right) = \left(\frac{R - \pitarget(X)}{\pitarget(X)} \right) h_2(X) \in \Lambda_2\,,
\end{equation}
where $h_2(X) = \EE{\shortfif\mid X}$. This concludes the proof.
\end{proof}

\subsection{Proof of Proposition~\ref{prop:corr}}
\begin{proof}
    Under Assumption~\ref{ass:MAR}, we have
    \begin{equation}
        \EE{\ipw} = \EE{\frac{R Y}{\hat\pi(X)}} = \EE{\frac{\pitarget(X) \mutarget(X)}{\hat\pi(X)}}\,.
    \end{equation}
    Then, given the fact that $\hat\mu =\mutarget$, we have
    \begin{equation}
        \EE{\ipw} = \EE{\frac{\pitarget(X) \mutarget(X)}{\hat\pi(X)}} = \EE{\frac{\pitarget(X) \hat\mu(X)}{\hat\pi(X)}} \,,
    \end{equation}
    which is the expected value of the quantity $n^{-1} \sum_{i=1}^n  R_i\hat\mu(X_i)/\hat\pi(X_i)$.

    Similarly, under Assumption~\ref{ass:MAR} and $\hat\pi=\pitarget$, we have
    \begin{equation}
        \EE{\oreg} = \EE{\hat\mu(X)} = \EE{\frac{\pitarget(X)\hat\mu(X)}{\pitarget(X)}} = \EE{\frac{R \hat\mu(X)}{\hat\pi(X)}}\,,
    \end{equation}
    which is the expected value of $n^{-1} \sum_{i=1}^n  R_i\hat\mu(X_i)/\hat\pi(X_i)$. These facts imply $\EE{\dr} = \thetatarget$.
\end{proof}

\subsection{Proof of Theorem~\ref{th:consistency}}
\begin{proof}
The proof relies on Lemma \ref{lemma:lim_clip}, Proposition \ref{prop:corr}, and the property of \texttt{DR} estimators. We know that, if a nuisance is well-specified, then $\dr \pto \thetatarget$. This can be written as
\begin{equation}
     \text{plim}\left(\dr\right) = \text{plim}\left(\oreg\right) + \text{plim}\left(\ipw\right) - \text{plim}\left(\hat C\right) = \thetatarget\,,
\end{equation}
where $\hat C = n^{-1} \sum_{i=1}^n R_i \hat\mu(X_i) / \hat\pi(X_i)$ denotes the correction term. Assume that $\hat\mu = \mutarget$ (the proof for $\hat\pi=\pitarget$ is identical). Then, by Proposition \ref{prop:corr}, we have
\begin{equation}
    \text{plim}\left(\hat C\right) = \text{plim}\left(\ipw\right) \,.
\end{equation}
By Lemma \ref{lemma:lim_clip}, we also know that
\begin{equation}
    \clip{\text{plim}\left(\hat C\right)} = \text{plim}\left(\hat C\right)\,,
\end{equation}
as the limit $\text{plim}\left(\ipw\right)$ is within the boundaries of the clipping operator. This implies
\begin{equation}
     \text{plim}\left(\acc\right) = \text{plim}\left(\oreg\right) + \text{plim}\left(\ipw\right) - \text{plim}\left(\clip{\hat C}\right) = \text{plim}\left(\dr\right) = \thetatarget\,,
\end{equation}
which delivers the result. 
\end{proof}

\subsection{Proof of Theorem~\ref{th:efficiency}}

\begin{proof}
    Denote $I_n = \left[\min\{\oreg,\ipw\} -\delta_n,\max\{\oreg,\ipw\} +\delta_n  \right]$. The \ACC{} estimator differs from \texttt{DR} strictly on the event that clipping occurs, denoted as $A_n = \{\dr \notin I_n\}$. Therefore, it suffices to prove that the probability of this event vanishes, i.e.~$\PP{A_n}\to0$. 

    For $\dr$ to fall outside the interval $I_{n}$ its distance to both $\oreg$ and $\ipw$ must exceed the slack threshold $\delta_{n}$. That is, we must have $|\dr -\oreg|>\delta_{n}$ and $|\dr -\ipw |>\delta_{n}$. This gives us the intersection bounding the clipping event:
    \begin{equation}
        A_{n}\subseteq\{|\dr-\oreg|>\delta_{n}\}\cap\{|\dr-\ipw|>\delta_{n}\}\,.
    \end{equation}
    Because the event requires both distances to exceed $\delta_{n}$, the probability of clipping is bounded by the probability of exceeding the distance to just one of the estimators:
    \begin{equation}
        \PP{A_n} \leq \min \left\{ \PP{|\dr -\oreg|>\delta_{n}}, \PP{|\dr -\ipw|>\delta_{n}} \right\}\,.
    \end{equation}
    Without loss of generality, assume $\oreg$ is the base estimator with the faster (or equal) convergence rate, meaning $\gamma=\alpha$. We analyze the distance $|\dr -\oreg|$ using the triangle inequality:
    \begin{equation}
        |\dr-\oreg|\leq|\dr-\thetatarget|+|\thetatarget-\oreg|\,.
    \end{equation}
    Substituting the assumed rates yields:
    \begin{equation}
         |\dr -\oreg|=O_{\mathbbmss{P}}(n^{-1/2})+O_{\mathbbmss{P}}(n^{-\gamma})=O_{\mathbbmss{P}}(n^{-\gamma})\,.
    \end{equation}
    The expression $|\dr-\oreg|=O_{\mathbbmss{P}}(n^{-\gamma})$ implies that for any $\varepsilon>0$, there exists an $M>0$ such that for sufficiently large $n$:
    \begin{equation}
        \PP{n^{\gamma}|\dr -\oreg|>M} < \varepsilon\,.
    \end{equation}

    By assumption, the slack sequence satisfies $\delta_{n}n^{\gamma}\to\infty$. Therefore, there exists an $N$ such that for all $n>N$, $\delta_{n}n^{\gamma}>M$. Consequently:
    \begin{equation}
        \PP{|\dr -\oreg|>\delta_{n}} = \PP{n^{\gamma}|\dr -\oreg|>\delta_{n}n^{\gamma}} \leq \PP{n^{\gamma}|\dr-\oreg|>M} < \varepsilon\,.
    \end{equation}

Since $\varepsilon$ is arbitrary, $\PP{|\dr-\oreg|>\delta_{n}}\to0$. Because $\PP{|\dr - \oreg|>\delta_n}\to0$, it immediately follows that $\PP{A_{n}}\to0$. On the complement event $A_{n}^{C}$ we have $\acc =\dr$. This directly yields $\acc-\dr =o_{\mathbbmss{P}}(n^{-1/2})$, proving the theorem.
\end{proof}

\subsection{Proof of Theorem~\ref{th:safety}}
\begin{proof}
Let $\hat C = n^{-1} \sum_{i=1}^n R_i \hat\mu(X_i) / \hat\pi(X_i)$. The \ACC{} estimator is defined as
\begin{equation}
    \acc = \oreg + \ipw - \clip{\hat C}\,.
\end{equation}
The proof proceeds by first establishing that the value of $\acc$ is always bounded by the values of $\min\{\oreg, \ipw\} - \delta_n$ and $\max\{\oreg, \ipw\} + \delta_n$. We consider three exhaustive cases for the value of $\hat C$:
\begin{enumerate}
    \item If $\hat C \le \min\{\oreg, \ipw\} - \delta_n$, the clipped value is $\min\{\oreg, \ipw\} - \delta_n$. The estimator becomes:
    \begin{equation}
        \acc = \oreg + \ipw - \left( \min\{\oreg, \ipw\} - \delta_n \right) = \max\{\oreg, \ipw\} + \delta_n \,.
    \end{equation}
    \item If $\hat C \ge \max\{\oreg, \ipw\} + \delta_n$, the clipped value is $\max\{\oreg, \ipw\} + \delta_n$. The estimator becomes:
    \begin{equation}
        \acc = \oreg + \ipw - \left( \max\{\oreg, \ipw\} + \delta_n \right) = \min\{\oreg, \ipw\} - \delta_n \,.
    \end{equation}
    \item If $\min\{\oreg, \ipw\} -\delta_n < \hat C < \max\{\oreg, \ipw\} +\delta_n$, the clip has no effect. The estimator is $\acc = \oreg + \ipw - \hat C$. As $\hat C$ is between $\min\{\oreg, \ipw\} - \delta_n$ and $\max\{\oreg, \ipw\} + \delta_n$, the value of $\oreg + \ipw - \hat C$ must also lie in the interval $\left[\min\{\oreg, \ipw\} - \delta_n , \max\{\oreg, \ipw\} + \delta_n\right]$.
\end{enumerate}
In all possible cases, the estimator value is guaranteed to be in the closed interval defined by its components:
\begin{equation}
    \min\{\oreg, \ipw\} -\delta_n \le \acc \le \max\{\oreg, \ipw\} + \delta_n \,.
\end{equation}
This establishes the key \enquote{interval property} of the estimator. This implies that the \ACC{} estimator is a convex combination of the \texttt{OR} and \texttt{IPW} estimators (up to $\delta_n$), that is
\begin{equation}
\label{eq:combination}
\begin{split}
    \acc &= \lambda \left(\min\{\oreg, \ipw\} -\delta_n\right) + \left( 1 - \lambda\right) \left(\max\{\oreg, \ipw\} + \delta_n\right) \\
    &= \lambda \min\{\oreg, \ipw\} + \left( 1 - \lambda\right) \max\{\oreg, \ipw\} + (1-2\lambda) \delta_n\,,
\end{split}
\end{equation}
where
\begin{equation}
    \lambda = \frac{\hat C - \left( \min \left\{\oreg, \ipw \right\} -  \delta_n \right)}{\left(\max \left\{\oreg, \ipw \right\} + \delta_n\right) - \left(\min \left\{\oreg, \ipw \right\} - \delta_n\right)} \in [0,1]\,.
\end{equation}
By defining $\Tilde{\lambda} = \lambda\onea{\oreg \leq \ipw} + (1-\lambda)\onea{\oreg > \ipw} \in \{0,1\}$, we can write Equation~\ref{eq:combination} as
\begin{equation}
     \acc =\Tilde\lambda \oreg  + \left( 1 - \Tilde\lambda\right) \ipw + (1-2\Tilde{\lambda})\delta_n\,.
\end{equation}

Next, we relate this property to the estimation error. By subtracting $\thetatarget$ on both sides of the previous Equation, taking absolute values, and by triangle inequality, we get
\begin{equation}
    |\acc - \thetatarget| \leq \hat\lambda |\oreg - \thetatarget| + \left( 1 -\hat\lambda\right) |\ipw - \thetatarget| + |1-2\Tilde{\lambda}|\delta_n\,.
\end{equation}
Finally, given $\Tilde{\lambda}\in[0,1]$, $|1-2\Tilde{\lambda}|\leq1$. Therefore:
\begin{equation}
    |\acc - \thetatarget| \leq \hat\lambda |\oreg - \thetatarget| + \left( 1 -\hat\lambda\right) |\ipw - \thetatarget| + \delta_n\,.
\end{equation}

We can also easily derive the maximum bound. For any point $x$ in an interval $[a, b]$ and any external point $c$, the distance $|x-c|$ is maximized at one of the endpoints, $a$ or $b$. Therefore, the absolute error of our estimator is bounded by the maximum of the absolute errors of the \texttt{OR} and \texttt{IPW} estimators for any given sample:
\begin{equation}
    |\acc - \thetatarget| \le \max\{|\oreg - \thetatarget|, |\ipw - \thetatarget|\} + \delta_n \,.
\end{equation}
This inequality shows that the error of the \ACC{} estimator is always less than or equal to the maximum of the errors of its components. 
\end{proof}

\subsection{Proof of Theorem~\ref{th:oracle}}
\begin{proof}
    We show the two bounds independently, starting with the first. Denote $L_n = \min\{\oreg,\ipw\} -\delta_n$, $U_n = \max\{\oreg,\ipw\} +\delta_n$, and $I_n = [L_n, U_n ]$. Partition the sample space into four mutually exclusive and exhaustive events:
    \begin{equation}
        \begin{split}
            E_{in} &= \{\thetatarget \in I_n\}, \\
            E_{out,mid} &= \{\thetatarget \notin I_n, \dr \in I_n\}, \\
            E_{out,opp} &= \{\thetatarget < L_n, \dr > U_n\} \cup \{\thetatarget > U_n, \dr < L_n\}, \\
            E_{out,same} &= \{\thetatarget < L_n, \dr < L_n\} \cup \{\thetatarget > U_n, \dr > U_n\}\,.
        \end{split}
    \end{equation}
    Aggregate the events in $E = E_{in} \cup E_{out,mid} \cup E_{out,opp}$ and $E^C = E_{out,same}$. 

    For each $j$, the elementary inequality $a \le b + (a - b)_+$ applied with $a = \ell^\star$ and $b = \ell(\hat\theta_j)$ gives $\ell^\star \le \ell(\hat\theta_j)+ (\ell^\star - \ell(\hat\theta_j))_+$ pointwise. Combining with Lemma~\ref{lem:pointwise}, we get
    \begin{equation}
        \ell(\acc) \leq \ell^\star \leq \ell(\hat\theta_j) + (\ell^\star - \ell(\hat\theta_j))_+\,.
    \end{equation}
Taking expectations gives $\risk{\acc} \leq \risk{\hat\theta_j} + \Pi_j$ -- the first bound.

We now move to the second bound. By Lemma~\ref{lem:base-Dn}, for each $j \in \{\texttt{OR}, \texttt{IPW}\}$,
\begin{equation}
    |\acc - \thetatarget| \leq |\hat\theta_j - \thetatarget| + |\oreg - \ipw| + \delta_n\,.
\end{equation}
Square both sides and take expectations, then square root. By Minkowski's inequality,
\begin{equation}
    \sqrt{\EE{(\acc- \thetatarget)^2}} \leq \sqrt{\EE{\left(|\hat\theta_j - \thetatarget| + |\oreg - \ipw| + \delta_n\right)^2}} \leq \sqrt{\EE{(\hat\theta_j - \thetatarget)^2}} + \sqrt{\EE{D^2_n}}\,,
\end{equation}
where $D_n = |\oreg - \ipw| + \delta_n$. Since $\EE{(\acc- \thetatarget)^2} = \risk{\acc}$ and similarly for $\hat\theta_j$, and $\EE{D^2_n} = V_n$, this reads
\begin{equation}
    \sqrt{\risk{\acc}} \leq \sqrt{\risk{\hat\theta_j}} + \sqrt{V_n}\,.
\end{equation}
Minimizing the right side over $j \in \{\texttt{OR}, \texttt{IPW}\}$, we get $\sqrt{\risk{\acc}} \leq \sqrt{\min\{\risk{\oreg}, \risk{\ipw}\}} + \sqrt{V_n}$.

Now, apply Young's inequality $(a+b)^2 \le (1+\eta)a^2 + (1 + 1/\eta)b^2$ with $a = \sqrt{\min\{\risk{\oreg}, \risk{\ipw}\}}$ and $b = \sqrt{V_n}$:
\begin{equation}
    \risk{\acc} \leq (1+\eta) \min\{\risk{\oreg}, \risk{\ipw}\} + \left( 1 + \frac{1}{\eta} \right) V_n\,. 
\end{equation}
The bound $V_n \leq 2 \EE{(\oreg - \ipw)^2} + 2\delta_n^2$ follows from $(x+y)^2 \le 2x^2 + 2y^2$ applied pointwise to $D_n = |\oreg - \ipw| + \delta_n$, then taking expectations.

Each of the bounds holds simultaneously and unconditionally. Since $\risk{\acc}$ is a fixed quantity, it is bounded by the minimum of all valid upper bounds.
\end{proof}

\section{Clipping with $\delta_n=0$}
\label{sec:delta0}

In Section~\ref{sec:acc}, we established the default recommendation for the slack sequence $\delta_n$, which ensures that adaptive correction clipping remains asymptotically inactive under correct nuisance model specification. While this choice prioritizes semiparametric efficiency, practitioners may instead wish to prioritize safety above all else. In this section, we analyze the special case where $\delta_n=0$.

Setting the slack term to zero enforces the strongest possible safety guarantee: the resulting \ACC{} estimate is strictly required to be a convex combination of its constituent Outcome Regression (\texttt{OR}) and Inverse Probability Weighting (\texttt{IPW}) estimators. This ``interval property'' ensures that the estimator can never stray outside the bounds defined by the simpler nuisance models, providing maximal protection against the extreme instabilities of double fragility.

However, this absolute safety comes with a theoretical trade-off. When $\delta_n=0$, the non-linear clipping operator may remain active even in large samples where the standard \texttt{DR} estimator is well-behaved. This leads to a non-Gaussian limiting distribution, which invalidates standard inferential procedures based on asymptotic normality. Below, we formally characterize the safety properties of this case and propose a parametric bootstrap procedure to conduct valid inference despite the non-normal asymptotic distribution.

\subsection{Safety}
First, here we show the stronger safety property that characterizes \ACC{} when $\delta_n=0$.

\begin{theorem}[Safety]
\label{th:safe}
    Assume the identifiability conditions described in Assumption~\ref{ass:MAR}. Define 
    \begin{equation}
        \hat\lambda = \frac{\oreg - \clip{\frac{1}{n}\sum_{i=1}^n \frac{R_i \hat\mu(X_i)}{\hat\pi(X_i)}}}{\oreg - \ipw} \in [0,1]\,.
    \end{equation}
    Then, we have
    \begin{equation}
    \begin{split}
        |\acc - \thetatarget| &\leq \hat\lambda |\oreg - \thetatarget| + \left( 1 -\hat\lambda\right) |\ipw - \thetatarget| \\
        &\leq \max\{|\oreg - \thetatarget|, |\ipw - \thetatarget|\} \,.
    \end{split}
    \end{equation}
\end{theorem}

\begin{proof}
Let $\hat C = n^{-1} \sum_{i=1}^n R_i \hat\mu(X_i) / \hat\pi(X_i)$. The \ACC{} estimator is defined as
\begin{equation}
    \acc = \oreg + \ipw - \clip{\hat C}\,.
\end{equation}
The proof proceeds by first establishing that the value of $\acc$ is always bounded by the values of $\oreg$ and $\ipw$. We consider three exhaustive cases for the value of $\hat C$:
\begin{enumerate}
    \item If $\hat C \le \min\{\oreg, \ipw\}$, the clipped value is $\min\{\oreg, \ipw\}$. The estimator becomes:
    \begin{equation}
        \acc = \oreg + \ipw - \min\{\oreg, \ipw\} = \max\{\oreg, \ipw\}\,.
    \end{equation}
    \item If $\hat C \ge \max\{\oreg, \ipw\}$, the clipped value is $\max\{\oreg, \ipw\}$. The estimator becomes:
    \begin{equation}
        \acc = \oreg + \ipw - \max\{\oreg, \ipw\} = \min\{\oreg, \ipw\}\,.
    \end{equation}
    \item If $\min\{\oreg, \ipw\} < \hat C < \max\{\oreg, \ipw\}$, the clip has no effect. The estimator is $\acc = \oreg + \ipw - \hat C$. As $\hat C$ is between $\oreg$ and $\ipw$, the value of $\oreg + \ipw - \hat C$ must also lie in the interval $[\oreg, \ipw]$.
\end{enumerate}
In all possible cases, the estimator value is guaranteed to be in the closed interval defined by its components:
\begin{equation}
    \min\{\oreg, \ipw\} \le \acc \le \max\{\oreg, \ipw\}\,.
\end{equation}
This establishes the key \enquote{interval property} of the estimator. This implies that the \ACC{} estimator is a convex combination of the \texttt{OR} and \texttt{IPW} estimators, that is
\begin{equation}
    \acc = \hat\lambda \oreg + \left( 1 -\hat\lambda\right) \ipw\,,
\end{equation}
where $\hat \lambda = (\oreg - \hat C)/(\oreg - \ipw) \in [0,1]$. 

Next, we relate this property to the estimation error. By subtracting $\thetatarget$ on both sides of the previous Equation, taking absolute values, and by triangle inequality, we get
\begin{equation}
    |\acc - \thetatarget| \le \hat\lambda |\oreg - \thetatarget| + \left( 1 -\hat\lambda\right) |\ipw - \thetatarget|\,.
\end{equation}
If $\hat\lambda$ is an independent estimate with respect to $\oreg$ and $\ipw$, then we can compute the bias of \ACC{}:
\begin{equation}
     \EE{|\acc - \thetatarget|} \le \EE{\hat\lambda} \EE{\left|\mutarget - \hat\mu\right|} + \left( 1 -\EE{\hat\lambda}\right) \EE{\left|\frac{\pitarget}{\hat\pi} - 1 \right| \left|\mutarget\right|}\,.
\end{equation}

We can also easily derive the maximum bound. For any point $x$ in an interval $[a, b]$ and any external point $c$, the distance $|x-c|$ is maximized at one of the endpoints, $a$ or $b$. Therefore, the absolute error of our estimator is bounded by the maximum of the absolute errors of the \texttt{OR} and \texttt{IPW} estimators for any given sample:
\begin{equation}
    |\acc - \thetatarget| \le \max\{|\oreg - \thetatarget|, |\ipw - \thetatarget|\}\,.
\end{equation}
This inequality shows that the in-sample error of the \ACC{} estimator is always less than or equal to the maximum of the in-sample errors of its components. Computing the absolute bias of the \texttt{OR} and \texttt{IPW} estimators, we arrive at the final result:
\begin{equation}
        \EE{\left|\acc - \thetatarget \right|} \leq \max\left\{\EE{\left|\mutarget - \hat\mu\right|}, \EE{\left|\frac{\pitarget}{\hat\pi} - 1 \right|\left|\mutarget\right|} \right\}\,.
\end{equation}
\end{proof}

\begin{remark}
    The previous result can be readily extended to a result involving the absolute bias of the estimator, by taking expectations. In this case, to simplify the analysis, it is enough to estimate $\hat\lambda$ on a separate independent sample with respect to $\oreg$ and $\ipw$, so that one gets
    \begin{equation}
    \begin{split}
        \EE{\left|\acc - \thetatarget \right|} &\leq \EE{\hat\lambda} \EE{\left|\hat\mu - \mutarget\right|} + \left(1-\EE{\hat\lambda}\right) \EE{\left|1 - \frac{\pitarget}{\hat\pi} \right| \left|\mutarget\right|} \\
        &\leq \max\left\{\EE{\left|\hat\mu - \mutarget\right|}, \EE{\left|1 - \frac{\pitarget}{\hat\pi} \right| \left|\mutarget\right|} \right\}\,.        
    \end{split}
    \end{equation}
\end{remark}

\subsection{Inference through parametric bootstrap}
The safety property of the \ACC{} estimator is achieved through the use of a non-linear clipping operator. A direct, unfortunate consequence of this is that, when $\delta_n=0$, the estimator asymptotic distribution is not normal, even when the nuisance models are correctly specified. Instead, the limiting distribution is a non-linear transformation of jointly normal random variables, which invalidates standard inferential procedures that rely on a normal approximation. However, we propose a \textit{parametric bootstrap} procedure to construct asymptotically valid confidence intervals. This method does not rely on a normal approximation; instead, it works by simulating the true, non-standard asymptotic distribution of our estimator and then using its empirical quantiles to form an interval.

Before presenting the next Theorem, which characterizes the asymptotic distribution of our \ACC{} estimator when $\delta_n=0$ and both nuisance models are well-specified, we introduce some additional notation. We define the scaled errors of \texttt{OR}, \texttt{IPW} and the nonclipped correction term, along with their limits in distribution, as:
\begin{equation}
\label{eq:asymp_exp}
    \begin{split}
        \sqrt{n}(\oreg - \thetatarget) &\dto \zoreg\,, \\
        \sqrt{n}(\ipw - \thetatarget) &\dto \zipw\,, \\
        \sqrt{n}\left(\frac{1}{n} \sum_{i=1}^n \frac{R_i\hat\mu(X_i)}{\hat\pi(X_i)} - \thetatarget \right) &\dto \zcorr\,,
    \end{split}
\end{equation}
where the vector $(\zoreg, \zipw, \zcorr)$ follows a multivariate normal distribution with mean zero and an unknown covariance matrix $\Sigma$.

\begin{theorem}[Asymptotic distribution of \ACC]
\label{th:asymp}
Assume the identifiability conditions described in Assumption~\ref{ass:MAR}. Assume also that the nuisance models, $\hat{\mu}$ and $\hat{\pi}$, are well-specified. Then
\begin{equation}
    \sqrt{n}(\acc - \thetatarget) \dto \zoreg + \zipw - \clip{\zcorr}\,.
\end{equation}
\end{theorem}

\begin{proof}
Let $\hat C = n^{-1} \sum_{i=1}^n R_i \hat\mu(X_i) / \hat\pi(X_i)$. Our goal is to find the limiting distribution of $\sqrt{n}(\acc - \thetatarget)$. We start by expressing the component estimators using an asymptotic expansion based on their limiting random variables:
\begin{equation}
    \begin{split}
        \oreg &= \thetatarget + \frac{\zoreg}{\sqrt{n}} + o_p(n^{-1/2})\,, \\
        \ipw &= \thetatarget + \frac{\zipw}{\sqrt{n}} + o_p(n^{-1/2})\,, \\
        \hat C &= \thetatarget + \frac{\zcorr}{\sqrt{n}} + o_p(n^{-1/2})\,.
    \end{split}
\end{equation}
The clipping bounds can be expanded similarly:
\begin{equation}
    \begin{split}
         L_n &= \min\{\oreg, \ipw\} = \thetatarget + \frac{1}{\sqrt{n}}\min\{\zoreg, \zipw\} + o_p(n^{-1/2})\,, \\
         U_n &= \max\{\oreg, \ipw\} = \thetatarget + \frac{1}{\sqrt{n}}\max\{\zoreg, \zipw\} + o_p(n^{-1/2})\,.
    \end{split}
\end{equation}
Now we analyze the clipped term. By Lemma \ref{lem:aec} we have:
\begin{equation}
    \clip{\hat C, L_n, U_n} = \thetatarget + \frac{1}{\sqrt{n}}\clip{\zcorr, \min\{\zoreg, \zipw\}, \max\{\zoreg, \zipw\}} + o_p(n^{-1/2})
\end{equation}
Finally, we substitute these expansions into the expression for the scaled error of the final estimator:
\begin{equation}
    \begin{split}
        \sqrt{n}(\acc - \thetatarget) &= \sqrt{n}\left( \oreg + \ipw - \clip{\hat C, L_n, U_n} - \thetatarget \right) \\
        &= \sqrt{n}\left( \left(\thetatarget + \frac{\zoreg}{\sqrt{n}}\right) + \left(\thetatarget + \frac{\zipw}{\sqrt{n}}\right) - \left(\thetatarget + \frac{\clip{\zcorr}}{\sqrt{n}}\right) - \thetatarget \right) + o_p(1) \\
        &= \sqrt{n}\left( \frac{\zoreg + \zipw -\clip{\zcorr}}{\sqrt{n}} \right) + o_p(1) \\
        &= \zoreg + \zipw -\clip{\zcorr, \min\{\zoreg, \zipw\}, \max\{\zoreg, \zipw\}} + o_p(1)\,.
    \end{split}
\end{equation}
As $n \to \infty$, the $o_p(1)$ term vanishes, which completes the proof.
\end{proof}

\begin{remark}
The Theorem recovers the consistency result of Theorem~\ref{th:consistency}, as the limiting distribution is centered correctly. However, the limiting distribution is not normal. In fact, the random variable $W = \zoreg + \zipw - \clip{\zcorr}$ is a non-linear transformation of jointly normal variables, and thus not normal. This is the reason why standard inferential procedures based on the normal approximation are not strictly valid for our estimator. Despite this, 
in our simulations
we find that the difference between our limit $W$ and a Gaussian distribution with proper variance is 
negligible (see Supplementary Figure \ref{fig:W_example} for an example). 
\end{remark}

\begin{algorithm}[t]
\caption{Parametric bootstrap}
\label{alg:pseudocode}
\begin{algorithmic}[1]
\Require $\alpha$ (confidence level), $B$ (number of bootstrap samples)
\Ensure Asymptotically valid confidence intervals for $\thetatarget$, i.e.~$[\acc -  q_{1-\alpha/2}/\sqrt{n}, \acc - q_{\alpha/2}/\sqrt{n} ]$
\State Compute estimators $\oreg$, $\ipw$, and $\hat C = \frac{1}{n} \sum_{i=1}^n \frac{R_i\hat\mu(X_i)}{\hat\pi(X_i)}$ 
\For{$i\leftarrow1\text{ to } n$}
\State Estimate empirical influence functions $\influence{\data_i; \oreg; \hat\mu}$,  $\influence{\data_i; \ipw; \hat\pi}$, and $\influence{\data_i; \hat C; \hat\mu; \hat\pi}$ 
\EndFor
\State Estimate covariance matrix, where each element is given by $\hat{\Sigma}_{jk} = \frac{1}{n} \sum_{i=1}^{n} \influence{\data_i;j} \influence{\data_i;k}$
\For{$b\leftarrow1\text{ to } B$}
\State Sample $\left(\zoreg^b,\zipw^b,\zcorr^b\right) \sim \Normal{0}{\hat\Sigma}$ and compute $W^b = \zoreg^b + \zipw^b - \clip{\zcorr^b}$
\EndFor
\State Find the empirical $\alpha/2$ and $1-\alpha/2$ quantiles of $\{W^b\}_{b=1}^B$, denoted $q_{\alpha/2}$ and $q_{1-\alpha/2}$
\end{algorithmic}
\end{algorithm}

Based on Theorem \ref{th:asymp}, we construct an asymptotically valid confidence interval by simulating the distribution of $W$ to
calculate its quantiles
(see Algorithm~\ref{alg:pseudocode}). 
This approach is, in essence, 
a parametric bootstrap. The procedure begins by computing the empirical influence functions for each of the three core components: the \texttt{OR} estimator, the \texttt{IPW} estimator, and the correction term. These influence functions, which capture the contribution of each observation to the variance, are then used to compute a consistent estimate, $\hat\Sigma$, of the joint asymptotic covariance matrix of the three components. Then, a large number ($B$) of random vectors are drawn from a multivariate normal distribution, $\Normal{0}{\hat\Sigma}$, to simulate the joint asymptotic behavior of the unclipped components, that is $\left(\zoreg^b,\zipw^b,\zcorr^b\right) \sim \Normal{0}{\hat\Sigma}$. For each $b=1,\dots,B$, we then compute $W^b = \zoreg^b + \zipw^b - \clip{\zcorr^b}$. This transforms the sample of jointly normal variables into a large empirical sample from the asymptotic distribution of the final estimator. From this simulated sample, the empirical $\alpha/2$ and $1-\alpha/2$ quantiles, denoted $q_{\alpha/2}$ and $q_{1-\alpha/2}$, are calculated and used to construct the asymptotically valid confidence interval around the point estimate, that is
\begin{equation}
    \left[\acc - \frac{q_{1-\alpha/2}}{\sqrt{n}}\,, \acc - \frac{q_{\alpha/2}}{\sqrt{n}} \right]\,.
\end{equation}

\newpage
\section{Additional simulation results}

\begin{figure}[ht]
    \centering
    \includegraphics[width=\linewidth]{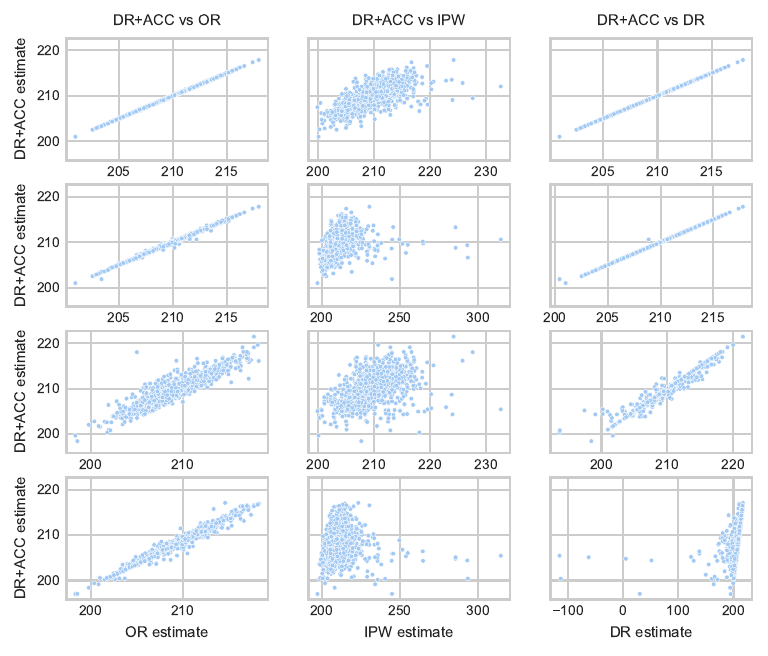}
    \caption{Scatterplots comparing the performance of the \ACC{} estimator (y-axis) against the \texttt{OR}, \texttt{IPW}, and standard \texttt{DR} estimators (x-axes) for a sample size of $n=200$. Each row corresponds to one of the four nuisance model specification scenarios (from top to bottom: correct $\hat\mu$/correct $\hat\pi$; correct $\hat\mu$/incorrect $\hat\pi$; incorrect $\hat\mu$/correct $\hat\pi$; incorrect $\hat\mu$/incorrect $\hat\pi$. The top three rows demonstrate that when at least one nuisance model is correct, the \ACC{} estimator closely tracks the standard \texttt{DR} estimator, which in turn tracks the well-specified estimator between \texttt{OR} and \texttt{IPW}, confirming its consistency. The bottom row illustrates the critical case of complete misspecification. The rightmost plot provides a direct visualization of the adaptive correction clipping mechanism: the \ACC{} estimate follows the standard \texttt{DR} when the latter is stable but remains bounded when the standard \texttt{DR} produces extreme, unstable estimates. This empirically demonstrates how the safety property of the \ACC{} estimator protects against the double fragility of the standard \texttt{DR} estimator. Full details of the simulation scenario are provided in Section \ref{sec:sim}.}
    \label{fig:scatter200}
\end{figure}

\begin{figure}[ht]
    \centering
    \includegraphics[width=\linewidth]{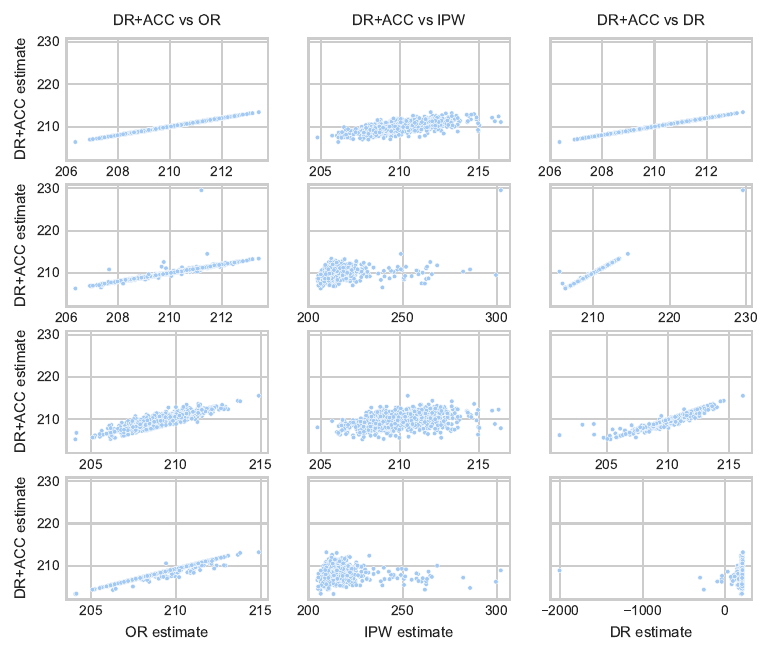}
    \caption{Scatterplots comparing the performance of the \ACC{} estimator (y-axis) against the \texttt{OR}, \texttt{IPW}, and standard \texttt{DR} estimators (x-axes) for a sample size of $n=1000$. Each row corresponds to one of the four nuisance model specification scenarios (from top to bottom: correct $\hat\mu$/correct $\hat\pi$; correct $\hat\mu$/incorrect $\hat\pi$; incorrect $\hat\mu$/correct $\hat\pi$; incorrect $\hat\mu$/incorrect $\hat\pi$. The top three rows demonstrate that when at least one nuisance model is correct, the \ACC{} estimator closely tracks the standard \texttt{DR} estimator, which in turn tracks the well-specified estimator between \texttt{OR} and \texttt{IPW}, confirming its consistency. The bottom row illustrates the critical case of complete misspecification. The rightmost plot provides a direct visualization of the adaptive correction clipping mechanism: the \ACC{} estimate follows the standard \texttt{DR} when the latter is stable but remains bounded when the standard \texttt{DR} produces extreme, unstable estimates. This empirically demonstrates how the safety property of the \ACC{} estimator protects against the double fragility of the standard \texttt{DR} estimator. Full details of the simulation scenario are provided in Section \ref{sec:sim}.}
    \label{fig:scatter1000}
\end{figure}

\begin{table}[ht]
\centering
\caption{Empirical coverage and confidence interval (CI) widths for a sample size of $n=100$, grouped by nuisance model specification. The nominal coverage level is 95\%.}
\label{tab:coverage_results_n100}
\begin{tabular}{llrr}
\toprule
\textbf{Scenario} & \textbf{Estimator} & \textbf{Coverage} & \textbf{CI Width} \\
\midrule
\multirow{4}{*}{Correct $\hat\mu$, Correct $\hat\pi$} & \texttt{OR} & 0.946 & 14.14 \\
 & \texttt{IPW} & 1.000 & 136.93 \\
 & \texttt{DR} & 0.950 & 14.16 \\
 & \ACC{} & 0.950 & 14.16 \\
\midrule
\multirow{4}{*}{Correct $\hat\mu$, Incorrect $\hat\pi$} & \texttt{OR} & 0.946 & 14.14 \\
 & \texttt{IPW} & 1.000 & 175.01 \\
 & \texttt{DR} & 0.949 & 14.18 \\
 & \ACC{} & 0.949 & 14.18 \\
\midrule
\multirow{4}{*}{Incorrect $\hat\mu$, Correct $\hat\pi$} & \texttt{OR} & 0.878 & 14.82 \\
 & \texttt{IPW} & 1.000 & 136.93 \\
 & \texttt{DR} & 0.957 & 19.20 \\
 & \ACC{} & 0.957 & 19.20 \\
\midrule
\multirow{4}{*}{Incorrect $\hat\mu$, Incorrect $\hat\pi$} & \texttt{OR} & 0.878 & 14.82 \\
 & \texttt{IPW} & 1.000 & 175.01 \\
 & \texttt{DR} & 0.928 & 23.20 \\
 & \ACC{} & 0.937 & 23.20 \\
\bottomrule
\end{tabular}
\end{table}

\begin{figure}[ht]
    \centering
    \includegraphics[width=\linewidth]{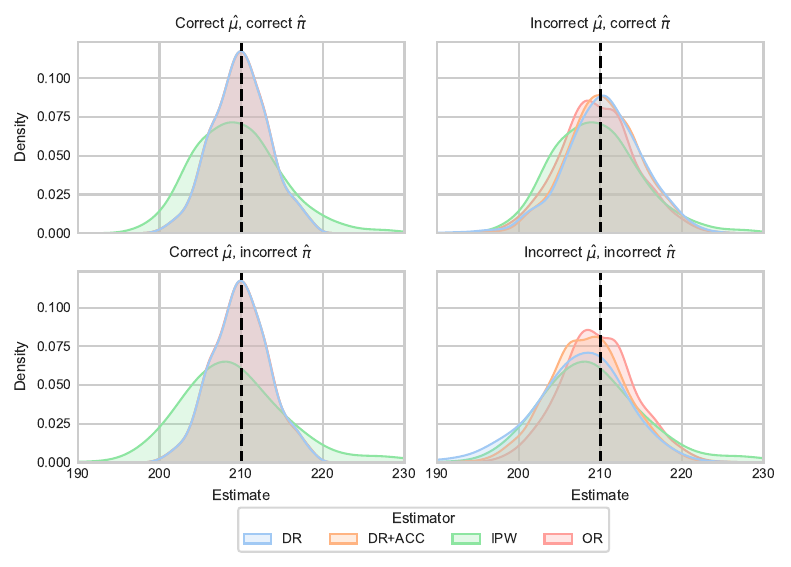}
    \caption{Sampling distributions for the Outcome Regression (\texttt{OR}), Inverse Probability Weighting (\texttt{IPW}), Doubly Robust (\texttt{DR}), and our proposed \ACC{} estimators from 1000 simulations with a sample size of $n=100$. The true parameter value is 210. The four panels show the estimators' performance under all combinations of correct and incorrect nuisance model specifications. The top row and bottom-left panel demonstrate the \textit{asymptotic hard thresholding} property. When at least one nuisance model is correct, both \texttt{DR} and \ACC{} align with the correct simpler estimator. The bottom-right panel illustrates \textit{double fragility}. When both nuisance models are wrong, the bias of the standard \texttt{DR} estimator can be worse than that of either the \texttt{OR} or \texttt{IPW} estimators. In this challenging scenario, our proposed \ACC{} estimator is shown to be \textit{safe}, providing a much more stable and accurate estimate than the standard \texttt{DR} estimator. Full details of the simulation scenario are provided in Section \ref{sec:sim}.}
    \label{fig:intro_plot_n100}
\end{figure}

\begin{figure}[ht]
    \centering
    \includegraphics[width=\linewidth]{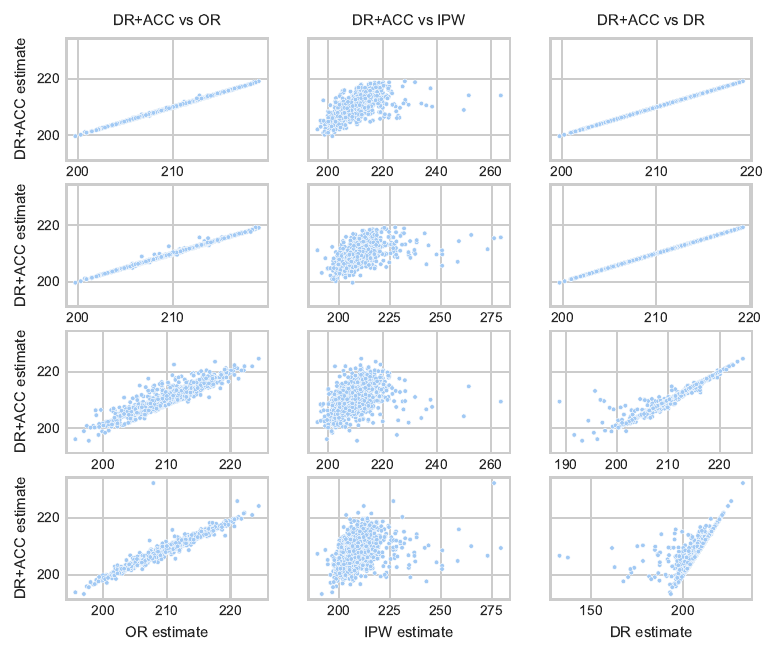}
    \caption{Scatterplots comparing the performance of the \ACC{} estimator (y-axis) against the \texttt{OR}, \texttt{IPW}, and standard \texttt{DR} estimators (x-axes) for a sample size of $n=100$. Each row corresponds to one of the four nuisance model specification scenarios (from top to bottom: correct $\hat\mu$/correct $\hat\pi$; correct $\hat\mu$/incorrect $\hat\pi$; incorrect $\hat\mu$/correct $\hat\pi$; incorrect $\hat\mu$/incorrect $\hat\pi$. The top three rows demonstrate that when at least one nuisance model is correct, the \ACC{} estimator closely tracks the standard \texttt{DR} estimator, which in turn tracks the well-specified estimator between \texttt{OR} and \texttt{IPW}, confirming its consistency. The bottom row illustrates the critical case of complete misspecification. The rightmost plot provides a direct visualization of the adaptive correction clipping mechanism: the \ACC{} estimate follows the standard \texttt{DR} when the latter is stable but remains bounded when the standard \texttt{DR} produces extreme, unstable estimates. This empirically demonstrates how the safety property of the \ACC{} estimator protects against the double fragility of the standard \texttt{DR} estimator. Full details of the simulation scenario are provided in Section \ref{sec:sim}.}
    \label{fig:scatter100}
\end{figure}

\begin{table}[ht]
\centering
\caption{Simulation results for $n=100$, $n=200$, and $n=1000$. This comparison highlights the performance of reweighted and stabilized estimators across increasing sample sizes.}
\label{tab:sim_results_refined}
\begin{tabular}{ll ccc ccc ccc}
\toprule
& & \multicolumn{3}{c}{\textbf{Bias}} & \multicolumn{3}{c}{\textbf{RMSE}} & \multicolumn{3}{c}{\textbf{MAE}} \\
\cmidrule(lr){3-5} \cmidrule(lr){6-8} \cmidrule(lr){9-11}
\textbf{Scenario} & \textbf{Estimator} & \textbf{100} & \textbf{200} & \textbf{1000} & \textbf{100} & \textbf{200} & \textbf{1000} & \textbf{100} & \textbf{200} & \textbf{1000} \\
\midrule
\multirow{6}{*}{\shortstack[l]{Correct $\hat\mu$,\\Correct $\hat\pi$}} 
 & \texttt{Hajek}                & -0.16 & -0.31 & -0.02 & 6.20 & 3.86 & 1.69 & 3.68 & 2.46 & 1.10 \\
 & \texttt{Trimmed IPW}          & -0.46 & -0.43 & -0.14 & 5.33 & 3.67 & 1.58 & 3.63 & 2.42 & 1.07 \\
 & \texttt{WLS}                  & -0.08 & -0.07 & -0.00 & 3.46 & 2.57 & 1.13 & 2.31 & 1.85 & 0.77 \\
 & \ACC{} ($\delta_n=0$) & -0.07 & -0.07 & -0.00 & 3.46 & 2.57 & 1.13 & 2.30 & 1.85 & 0.77 \\
 & \texttt{\ACC{}}               & -0.07 & -0.07 & -0.00 & 3.46 & 2.57 & 1.13 & 2.30 & 1.85 & 0.77 \\
 & \texttt{TMLE}                 & -0.08 & -0.07 & -0.00 & 3.46 & 2.57 & 1.13 & 2.31 & 1.85 & 0.77 \\
\midrule
\multirow{6}{*}{\shortstack[l]{Correct $\hat\mu$,\\Incorrect $\hat\pi$}} 
 & \texttt{Hajek}                & 0.03  & 1.59  & 4.76  & 8.80 & 9.73 & 11.10 & 4.24 & 3.41 & 2.56 \\
 & \texttt{Trimmed IPW}          & -1.07 & -0.59 & 0.24  & 5.93 & 4.13 & 1.91  & 4.06 & 2.79 & 1.33 \\
 & \texttt{WLS}                  & -0.07 & -0.07 & -0.00 & 3.46 & 2.57 & 1.13  & 2.29 & 1.86 & 0.77 \\
 & \ACC{} ($\delta_n=0$) & 0.13  & -0.26 & -0.39 & 4.53 & 3.20 & 1.54  & 2.93 & 2.13 & 1.08 \\
 & \texttt{\ACC{}}               & -0.06 & -0.07 & 0.03  & 3.46 & 2.57 & 1.30  & 2.30 & 1.84 & 0.78 \\
 & \texttt{TMLE}                 & -0.11 & -0.08 & -0.00 & 3.46 & 2.57 & 1.13  & 2.28 & 1.85 & 0.77 \\
\midrule
\multirow{6}{*}{\shortstack[l]{Incorrect $\hat\mu$,\\Correct $\hat\pi$}} 
 & \texttt{Hajek}                & -0.16 & -0.31 & -0.02 & 6.20 & 3.86 & 1.69 & 3.68 & 2.46 & 1.10 \\
 & \texttt{Trimmed IPW}          & -0.46 & -0.43 & -0.14 & 5.33 & 3.67 & 1.58 & 3.63 & 2.42 & 1.07 \\
 & \texttt{WLS}                  & 0.62  & 0.30  & 0.13  & 4.23 & 3.04 & 1.35 & 2.69 & 2.05 & 0.92 \\
 & \ACC{} ($\delta_n=0$) & -0.05 & -0.05 & 0.06  & 3.47 & 2.57 & 1.30  & 2.30 & 1.85 & 0.77 \\
 & \texttt{\ACC{}}               & 0.45  & 0.08  & -0.15 & 4.51 & 3.24 & 1.52  & 2.98 & 2.19 & 1.05 \\
 & \texttt{TMLE}                 & 0.00  & -0.05 & -0.05 & 4.29 & 3.13 & 1.46  & 2.73 & 2.09 & 1.01 \\
\midrule
\multirow{6}{*}{\shortstack[l]{Incorrect $\hat\mu$,\\Incorrect $\hat\pi$}} 
 & \texttt{Hajek}                & 0.03  & 1.59  & 4.76  & 8.80 & 9.73 & 11.10 & 4.24 & 3.41 & 2.56 \\
 & \texttt{Trimmed IPW}          & -1.07 & -0.59 & 0.24  & 5.93 & 4.13 & 1.91  & 4.06 & 2.79 & 1.33 \\
 & \texttt{WLS}                  & -1.73 & -2.37 & -2.96 & 4.86 & 3.97 & 3.29  & 3.30 & 2.84 & 2.95 \\
 & \ACC{} ($\delta_n=0$) & -0.80 & -1.09 & -0.91 & 4.67 & 3.38 & 1.72  & 2.93 & 2.26 & 1.23 \\
 & \texttt{\ACC{}}               & -1.53 & -1.96 & -1.60 & 5.00 & 3.81 & 2.16  & 3.42 & 2.62 & 1.66 \\
 & \texttt{TMLE}                 & -3.60 & -4.00 & -4.74 & 6.41 & 5.63 & 5.30  & 4.08 & 4.02 & 4.38 \\
\bottomrule
\end{tabular}
\end{table}

\begin{figure}[ht]
    \centering
    \includegraphics[width=\linewidth]{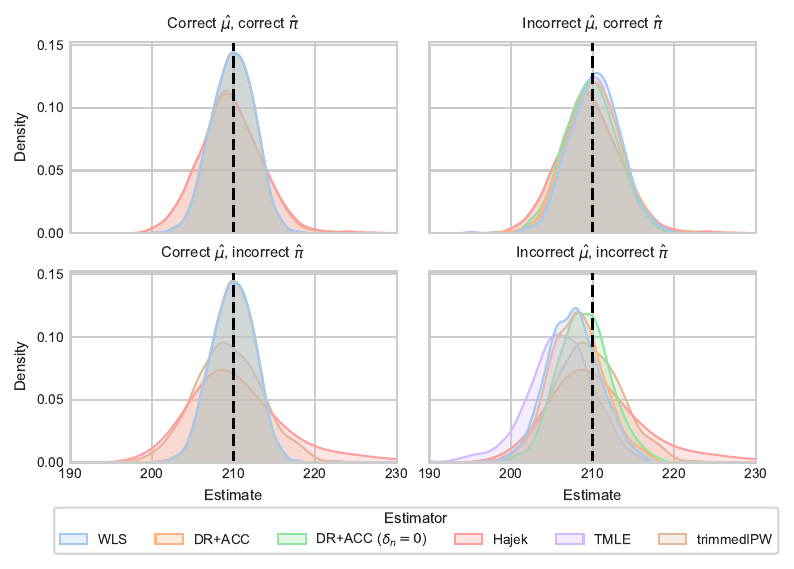}
    \caption{Sampling distributions for Hajek Inverse Probability Weighting (\texttt{Hajek}), trimmed Inverse Probability Weighting (\texttt{trimmed IPW}), weighted least squares (\texttt{WLS}), our proposed \ACC{}, our proposed \ACC{} with $\delta_n=0$, and targeted maximum likelihood (\texttt{TMLE}) \citep{van2006targeted, frank2024implementing} estimators from 1000 simulations with a sample size of $n=200$. The true parameter value is 210. The four panels show the estimators' performance under all combinations of correct and incorrect nuisance model specifications.  Full details of the simulation scenario are provided in Section \ref{sec:sim}.}
    \label{fig:intro_plot_n200_extra}
\end{figure}

\begin{figure}[ht]
    \centering
    \includegraphics[width=\linewidth]{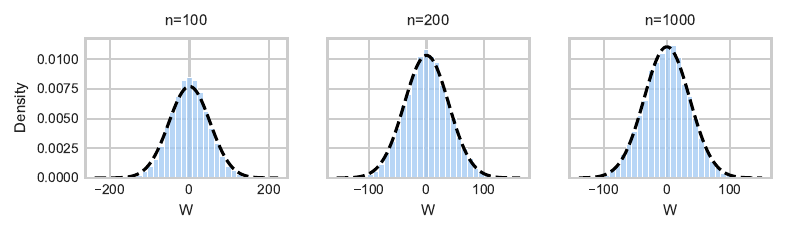}
    \caption{Comparison between empirical distribution of $W$ and Gaussian approximation of \texttt{DR} (black dashed line) for a random simulation draw, when both nuisance models are well-specified. Full details of the simulation scenario are provided in Section \ref{sec:sim}.}
    \label{fig:W_example}
\end{figure}

\clearpage
\break
\section{Additional application details}
\begin{table}[ht]
\centering
\caption{List of significant peptides and their related genes. Significance level 0.05. No multiple testing correction applied. We retrieve the UniProt Accession Number from the peptide string and use it to query the UniProt database API to get primary gene information.}
\label{tab:gene_list}
\begin{tabular}{lll}
\toprule
\textbf{Peptide} & \textbf{Gene name} & \textbf{Reference} \\
\midrule
\texttt{AGTPYTVTLHGEVR} & \textit{TNC} &  \citet{xie2013tenascin}     \\
\texttt{AVEEEDKMTPEQLAIK} & \textit{PSMD14} & \citet{liu2024alzheimer}        \\
\texttt{DAEVERDEER} & \textit{MYH14} & \citet{finsterer2019neuropathy}\\
\texttt{DAVTYTEHAKR} & \textit{H4C1} & \citet{silvestro2021moringin}  \\
\texttt{ETQEDKLEGGAAK} & \textit{EPB41L2} & \citet{huang2024integrative} \\
\texttt{GDEEEEGEEKLEEK} & \textit{CANX} & \citet{shen2025er}              \\
\texttt{IPTHLFTFIQFK} & \textit{RO60} & \citet{crooke2022reduced}      \\
\texttt{LKHEC[+57]GAAFTSK} & \textit{CUL4A} & \citet{yasukawa2020nrbp1}      \\
\texttt{MGVAAHKK} & \textit{S100A8} & \citet{litus2025binding}        \\
\texttt{SLSTSGESLYHVLGLDK} & \textit{DNAJC5} & \citet{rosene2023cell}          \\
\texttt{TAEHEAAQQDLQSK} & \textit{KTN1} & \citet{li2024identification}   \\
\texttt{TANKDHLVTAYNHLFETK} & \textit{FKBP3} & \citet{blair2015emerging}       \\
\texttt{TNNVSEHEDTDKYR} & \textit{COPB1} & \citet{yang2019genetic}         \\
\bottomrule
\end{tabular}
\end{table}

\begin{figure}[ht]
    \centering
    \includegraphics[width=\linewidth]{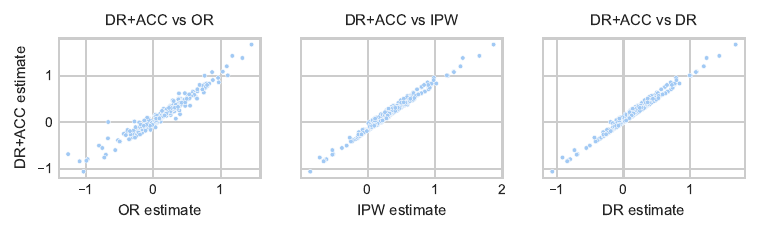}
    \caption{Scatterplots comparing the estimates of the \ACC{} estimator (y-axis) against the \texttt{OR}, \texttt{IPW}, and standard \texttt{DR} estimators (x-axes) for the ATE associated to each of the 270 peptides. Full details of the application study are provided in Section \ref{sec:app}.}
    \label{fig:app_scatter}
\end{figure}

\end{document}